\documentclass[11pt, svgnames]{article}
\usepackage[english]{babel}
\usepackage[margin=1 in]{geometry}
\usepackage[utf8]{inputenc}
\usepackage{url}
\usepackage{amsmath,amssymb,amsthm,mathrsfs,amsfonts,amscd,keyval}
\usepackage{float}
\usepackage{mathtools,mathrsfs}
\usepackage{setspace}
\usepackage{multirow}
\usepackage{booktabs,longtable,array}
\usepackage{algorithm}
\usepackage{algpseudocode}
\usepackage[svgnames]{xcolor}
\usepackage{enumitem}
\usepackage{hyperref}
\usepackage{parskip}
\hypersetup{
colorlinks=true,
linkcolor=blue,
filecolor=blue,
citecolor=blue,  
urlcolor=blue,
linktocpage=true}
\usepackage{authblk}
\usepackage{placeins}
\newcommand{\ket}[1]{\lvert #1\rangle}
\usepackage{tikz}
\usetikzlibrary{positioning,fit,arrows.meta,calc,backgrounds}
\theoremstyle{definition}
\usepackage{mathtools}
\usepackage{subcaption}
\def\iprintdate#1/#2/#3/{\DTMdisplaydate{#1}{#2}{#3}{-1}}

\newtheorem{theorem}{Theorem}[section]
\newtheorem{lemma}[theorem]{Lemma}
\newtheorem{definition}[theorem]{Definition}
\newtheorem{example}[theorem]{Example}

\theoremstyle{remark}
\newtheorem{remark}[theorem]{Remark}
\theoremstyle{definition}
\newtheorem{procedure}[theorem]{Procedure}

\usepackage{amsmath,amssymb}
\usepackage{xcolor}
\usepackage{graphicx}
\usepackage{tabularx} % preamble
\usepackage{tikz}
\usetikzlibrary{arrows.meta,positioning,calc}

\definecolor{algink}{HTML}{1F2933}
\definecolor{algmuted}{HTML}{66788A}
\definecolor{algnavy}{HTML}{173B63}
\definecolor{algblue}{HTML}{276EA7}
\definecolor{algbluepale}{HTML}{EAF3FA}
\definecolor{algslatepale}{HTML}{F4F6F8}
\definecolor{alglinegray}{HTML}{C7D1DB}
\definecolor{algamber}{HTML}{A65300}

\newcommand{\Ftwo}{\mathbb F_{2}}
\newcommand{\symp}{\omega}
\newcommand{\orth}[1]{#1^{\perp}}
\newcommand{\dual}[1]{#1^{*}}
\newcommand{\atcut}[2]{#1^{(#2)}}

\newcommand{\FinalCut}{T}
\newcommand{\Vcut}{\mathsf{V}}
\newcommand{\VST}{\mathsf{V}_{\mathrm{ST}}}
\newcommand{\Stab}{\mathsf{S}}
\newcommand{\Gspace}{\mathsf{G}}

\newcommand{\DetectorSpace}{\mathsf{D}}
\newcommand{\Sdet}{\mathsf{S}_{\mathrm{det}}}
\newcommand{\Logical}{\mathsf{L}}

\newcommand{\CircuitBlock}{Q}

\newcommand{\ms}{\mathcal{M}}
\newcommand{\dmeas}{d_{\mathrm{meas}}}
\newcommand{\DetElt}{\delta}
\newcommand{\DetResp}{\sigma}

\newcommand{\DetMap}{\iota_{\mathrm{det}}}
\newcommand{\DetStab}[1]{s_{#1}}

\newcommand{\Uspace}{\mathsf{U}}
\newcommand{\Tspace}{\mathsf{T}}
\newcommand{\Pspace}{\mathsf{P}}
\newcommand{\Uelt}[2]{u_{#1}^{(#2)}}
\newcommand{\Telt}[2]{\tau_{#1}^{(#2)}}
\newcommand{\Pelt}[2]{p_{#1}^{(#2)}}
\newcommand{\Udual}[2]{u_{#1}^{\sharp(#2)}}
\newcommand{\Tdual}[2]{\tau_{#1}^{\sharp(#2)}}
\newcommand{\Pdual}[2]{p_{#1}^{\sharp(#2)}}

\newcommand{\Tsharp}{\mathsf{T}^{\sharp}}
\newcommand{\Psharp}{\mathsf{P}^{\sharp}}

\newcommand{\Ladj}{\mathsf{L}_{\mathrm{adj}}}
\newcommand{\Rprop}{\mathsf{R}}
\newcommand{\Rcomm}{\mathsf{R}_{\mathrm{comm}}}
\newcommand{\Ranti}{\mathsf{R}_{\mathrm{anti}}}
\newcommand{\Z}[1]{Z\!\left(#1\right)}
\newcommand{\Lbare}[1]{\mathsf{L}_{\mathrm{bare}}\!\left(#1\right)}
\newcommand{\Ldressed}[1]{\mathsf{L}_{\mathrm{dressed}}\!\left(#1\right)}
\newcommand{\CodeDistance}[1]{d_{#1}}

\begin{document}
\title{Subsystem Spacetime Code}

\author[1]{Esther Xiaozhen Fu}
\author[1]{Aranya Chakraborty}
\author[1]{Daniel Gottesman}
\affil[1]{Joint Center for Quantum Information and Computer Science,
University of Maryland, College Park, MD, USA}

\date{}

\maketitle

\begin{abstract}

Quantum error-correcting codes provide a powerful framework for analyzing and designing fault-tolerant Clifford circuits. In this work, we construct the subsystem spacetime code and use it to characterize the distances of Clifford circuits. We show that its gauge group can be explicitly constructed and classified using the masking properties of the instantaneous stabilizer groups \cite{fu2024errorcorrectiondynamicalcodes} and consists precisely of faults that are harmless to the circuit.  
We use this framework to identify where a dynamical code loses distance and to analyze how the fault distance changes when circuits are composed from smaller components. Our results provide a versatile tool for analyzing and designing fault-tolerant circuits, with potential applications to the construction of new dynamical codes.
\end{abstract}
\tableofcontents
\section{Introduction}
\label{sec:Introduction}

Quantum computation is inherently susceptible to noise, and building fault-tolerant
protocols with low resource overhead remains a central challenge for scalable quantum
computing, both in theory and in experiment.

The traditional approach treats a quantum code as a static object and builds circuits
from fault-tolerant gadgets, each implementing one task. A fault-tolerant gate, for
example, may consist of a transversal logical gate followed by an error-correction
gadget. This decomposition also underlies the standard threshold
analyses~\cite{aliferis2006threshold,aharonov1999faulttolerant}, in which each
component is designed so that residual errors that exit the previous gadget do not spread uncontrollably in the next gadget.

However, the traditional approach precludes a variety of fault-tolerant protocols. A broader view of fault tolerance allows a much wider variety of fault-tolerant circuits~\cite{gottesman2022opportunities}.  For instance, fault-tolerant protocols could integrate computation and error correction, change error-correcting codes frequently, or allow errors to spread by tracking how they change over time.  Relaxing the constraints on fault tolerance could facilitate improvements in the efficiency of fault tolerance.

This view has led to a dynamical approach to quantum codes, in which the stabilizer group
evolves under unitary gates and Pauli measurements. In a dynamical
code~\cite{Hastings2021dynamically,fu2024errorcorrectiondynamicalcodes}, the
instantaneous stabilizer group need not be fixed, measurements need not commute with
it, and syndrome extraction can proceed together with logical operations. Examples
include Floquet codes, code switching, and sequences of non-commuting measurements
that together implement a logical operation. The search for new dynamical codes has
produced many new
constructions~\cite{Anonymous_2025,Townsend_Teague_2023,PhysRevB.108.205116,PRXQuantum.4.020341,PRXQuantum.5.010342,PRXQuantum.5.020305,yan2024floquetcodescoupledspin,PhysRevB.108.195134,PRXQuantum.6.010360,PRXQuantum.5.040327,tanggara2025simpleconstructionquditfloquet,Hastings2021dynamically,PhysRevA.111.032413,Davydova2024quantumcomputation}.
A complementary line of work uses the ZX calculus to rewrite circuits while preserving
their fault-tolerance properties. This relates apparently different fault-tolerant
circuits and gives a systematic way to transform and simplify existing Clifford
circuits~\cite{delafuente2024dynamicalweightreductionpauli,Bombin2024unifyingflavorsof,McEwen_2023,Gidney_2023,Townsend_Teague_2023,PhysRevLett.134.090602,xu2025faulttolerantprotocolsspacetimeconcatenation, RodatzPoorKissinger2026}.

These developments call for an algebraic framework for the fault distance of such
circuits. For circuits of Clifford gates and Pauli measurements, the distances of the
instantaneous stabilizer groups do not determine the fault distance. A circuit can consist of good instantaneous codes but still have a small distance: measurement errors
can be undetectable and cause a logical error, and the stabilizers may not persist long enough to detect an error. The relevant distance is a property of
the whole spacetime circuit instead of that of an instantaneous stabilizer group. We would like a code that captures this distance exactly,
and whose structure shows how it relates to the familiar distances as well as the masking properties of the instantaneous codes.

Spacetime descriptions of quantum circuits have appeared in several earlier works.
Bacon, Flammia, Harrow and Shi~\cite{bacon_sparse_2017} constructed subsystem codes
from quantum circuits, with gauge operators accounting for errors that can be moved
through the circuit, and showed that a fault-tolerant circuit yields a spacetime code
with a corresponding distance guarantee. Their construction, however, assumes that
every measurement commutes with the current stabilizer group. Delfosse and
Paetznick~\cite{delfosse2023spacetimecodescliffordcircuits} defined the spacetime code
of a general Clifford circuit, allowing measurements that anticommute with the current
stabilizers, and used it to reduce circuit decoding to the decoding of a stabilizer
code. Their spacetime code is a stabilizer code and does not have gauge degrees
of freedom. In Ref.~\cite{gottesman2022opportunities}, both gauge operators and
anticommuting measurements are considered, but only when the measurements are made at
the final time step of the protocol. Blackwell and
Haah~\cite{blackwell2025code} studied the distance of Floquet codes in the steady
state. They identified a subspace of ``benign'' undetectable spacetime errors and
proved that the code distance is the minimum weight of an undetectable spacetime error
that is not benign. Their benign errors play the role that the gauge group plays here. Concurrent with this work, Xu, Wang and Liu~\cite{xu2026theoryspacetimequantumfault} independently obtained a subsystem spacetime code from the tensor-network representation of a Clifford circuit and used it to analyze fault-tolerant gadgets for static codes.   

Our work differs in the following aspects. First, we construct the subsystem spacetime gauge group
explicitly from the masking properties of the instantaneous stabilizer groups and
characterize its full structure (see Figure~\ref{fig:ss-ingredients}). Importantly, the masking structure makes the code usable for analyzing circuits: It gives the origin of every gauge qubit through its masking properties. It bounds the fault distance by the unmasked distances of the instantaneous codes and a distance for measurement errors, with equality under a correctability condition, which locates where a circuit loses distance. This also establishes the relation between the distances of the spacetime code and the unmasked distances of the instantaneous stabilizer groups, identifying conditions under which these distances agree. Our framework also provides a way to study the composition of fault-tolerant circuits by tracking how stabilizer information is preserved, recovered, or lost when individual circuit components are combined, and how this affects the distance of the resulting circuit. Finally, our framework accommodates different circuit settings, including state preparation, incomplete terminal measurements, logical
measurements, qubit reset and discard, and different fault models for measurements. 

\definecolor{sscDet}{HTML}{EB6123}
\definecolor{sscUnm}{HTML}{EB6123}
\definecolor{sscTmp}{HTML}{EB6123}
\definecolor{sscPer}{HTML}{EB6123}
\definecolor{sscLog}{HTML}{512888}
\definecolor{sscBg}{HTML}{F4EFF8}
\definecolor{sscRcomm}{HTML}{512888}

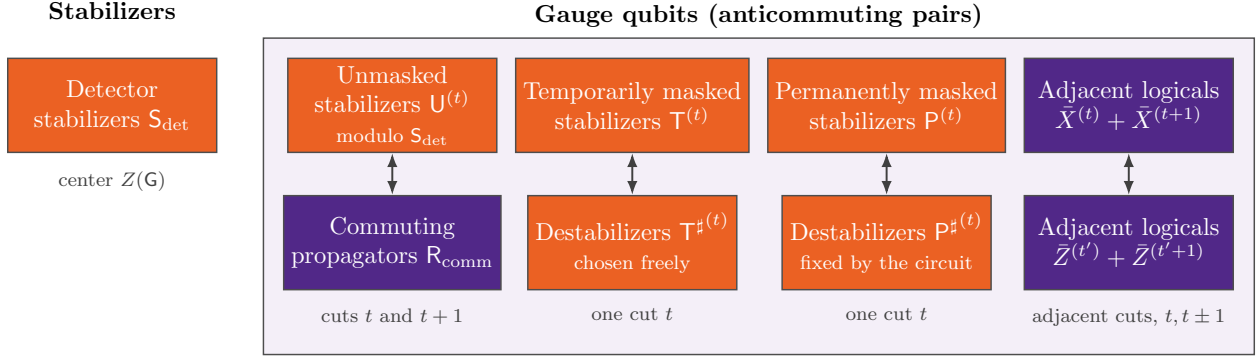
\begin{figure}[ht]
\centering
\resizebox{\linewidth}{!}{%
\begin{tikzpicture}[
    font=\small,
    box/.style={draw=black!70, line width=0.6pt, text=white, align=center,
                minimum width=3.0cm, minimum height=1.35cm, inner sep=3pt},
    pair/.style={{Latex[length=2mm]}-{Latex[length=2mm]}, line width=0.8pt, black!75},
    lab/.style={font=\scriptsize, text=black!75, align=center},
    node distance=0.6cm and 0.25cm]
  % stabilizers = center
  \node[box, fill=sscDet] (det) {Detector\\stabilizers $\Sdet$};
  \node[lab, below=0.12cm of det] {center $\Z{\Gspace}$};
  \node[font=\small\bfseries, above=0.42cm of det] {Stabilizers};

  % gauge qubits: top row and their partners in the bottom row
  \node[box, fill=sscUnm, right=1.0cm of det] (u)
      {Unmasked\\stabilizers $\atcut{\Uspace}{t}$\\{\scriptsize modulo $\Sdet$}};
  \node[box, fill=sscTmp, right=of u] (tm)
      {Temporarily masked\\stabilizers $\atcut{\Tspace}{t}$};
  \node[box, fill=sscPer, right=of tm] (pm)
      {Permanently masked\\stabilizers $\atcut{\Pspace}{t}$};
  \node[box, fill=sscLog, right=of pm] (l1)
      {Adjacent logicals\\$\bar X^{(t)}+\bar X^{(t+1)}$};

  \node[box, fill=sscRcomm, below=of u] (r)
      {Commuting\\propagators $\Rcomm$};
  \node[box, fill=sscTmp, below=of tm] (ts)
      {Destabilizers ${\Tsharp}^{(t)}$\\{\scriptsize chosen freely}};
  \node[box, fill=sscPer, below=of pm] (ps)
      {Destabilizers ${\Psharp}^{(t)}$\\{\scriptsize fixed by the circuit}};
  \node[box, fill=sscLog, below=of l1] (l2)
      {Adjacent logicals\\$\bar Z^{(t')}+\bar Z^{(t'+1)}$};

  \foreach \a/\b in {u/r, tm/ts, pm/ps, l1/l2} \draw[pair] (\a) -- (\b);

  % support of each family
  \node[lab, below=0.12cm of r]  (s1) {cuts $t$ and $t+1$};
  \node[lab, below=0.12cm of ts] (s2) {one cut $t$};
  \node[lab, below=0.12cm of ps] (s3) {one cut $t$};
  \node[lab, below=0.12cm of l2] (s4) {adjacent cuts, $t, t\pm1$};

  \begin{scope}[on background layer]
    \node[draw=black!70, line width=0.6pt, fill=sscBg, inner sep=0.28cm,
          fit=(u)(l1)(r)(l2)(s1)(s4),
          label={[font=\small\bfseries]above:Gauge qubits (anticommuting pairs)}] {};
  \end{scope}
\end{tikzpicture}}
\caption{Key ingredients of the subsystem spacetime code. The detector stabilizers
$\Sdet$ form the center of $\Gspace$ and are its stabilizers. Every other generator
belongs to a gauge qubit: a pair of operators that anticommute with each other
and commute with every other pair. Unmasked stabilizers, taken modulo
$\Sdet$, pair with commuting propagators across adjacent cuts; temporarily and
permanently masked stabilizers pair with their destabilizers at the same cut; and
adjacent logical generators pair with each other. The diagram shows a basis for each type of elements, but the exact gauge qubits, i.e. stabilizer-destabilizer pairs, need to be calculated from here. In our construction the elements in \textcolor{sscDet!80!white}{orange} come from the instantaneous gauge group (Def.~\ref{def : IGG}) whereas the elements in \textcolor{sscRcomm!80!white}{purple} come from the commuting propagators (Def.~\ref{def:comm_prop_space}).}
\label{fig:ss-ingredients}
\end{figure}

We now state our main contributions. Precise definitions are given in Sections~\ref{sec:spacetime-setup} and~\ref{sec: SS code}. We classify stabilizers based on masking properties. A stabilizer is
\emph{unmasked} if later measurement outcomes reveal its value. It is \emph{permanently masked}
if a later measurement destroys it before its value is ever read, and \emph{temporarily masked}
if its value has not been read by the end of the circuit but could still be read by further measurements.

\begin{itemize}
    \item \textbf{Subsystem spacetime code construction}: For a circuit of Clifford unitaries and Pauli
    measurements, the gauge group of the subsystem spacetime code is generated by the
    \emph{instantaneous gauge groups} (see Definition~\ref{def : IGG}) at each cut, which are determined by the masking
    properties of the instantaneous stabilizer groups, together with the commuting
    propagators (see Definition~\ref{def:propagators}). 

    \item \textbf{Characterization of the gauge group}: Figure~\ref{fig:ss-ingredients} shows the key ingredients through the lens of masking. The center of the gauge group is exactly the space of
    stabilizers corresponding to the detectors. The gauge qubits fall into three categories (see Section~\ref{sec:spacetime-gauge-construction}): masked stabilizers
    paired with their destabilizers, unmasked stabilizers paired with commuting
    propagators, and pairs formed from adjacent logical representatives. We give
    explicit algorithms that construct these pairs. The distance of the subsystem spacetime code equals the fault distance of the circuit (see proof in Section~\ref{sec:spacetime-logical-distance}),
    and maximum-likelihood decoding of the circuit reduces to coset decoding of the code. For simplicity, we construct the code with errors only on data qubits, but show in Subsection~\ref{subsec: Frameworks for measurements} that the framework easily includes other fault models, including the model underlying standard circuit-level distance.

    \item \textbf{Relation to the unmasked distance.} In general, the distance of a subsystem spacetime code is upper bounded by its unmasked distances. For measurement-only circuits
    satisfying a correctability condition (see Condition 1 in Section~\ref{sec:two-distances}), which holds for all static stabilizer codes and some dynamical codes, the spacetime distance equals
    $\min\{d_U,d_{\mathrm{meas}}\}$ (see Theorem~\ref{thm:two-distances}), where $d_U$ is the minimum unmasked distance of the
    instantaneous stabilizer groups and $d_{\mathrm{meas}}$ is a distance for
    measurement errors. The \emph{measurement distance} $d_{\mathrm{meas}}$ is the smallest weight of an undetectable but non-trivial set of measurement errors. We give an explicit example showing that this equality can fail without the correctability condition in Appendix~\ref{sec: counter}.

\item \textbf{Circuit composition.} Gluing two circuits whose distances have been computed separately can create logical faults that neither piece has, and these faults lie near the boundary where the two pieces meet. Under a correctability condition (see Condition 2 in Section~\ref{sec: circuit composition}), the composed distance is $\min\{d_1(C),d_2,d_{\mathrm{meas}}\}$ where  $d_1,d_2$ are the distances of the two pieces, $d_1(C)$ is the new distance of $d_1$ with a reclassification of temporarily masked stabilizers into unmasked and permanently masked stabilizers and the quantity $d_{\mathrm{meas}}$ is computed for the composed circuit $C$. Without this condition, we can exactly compute the composed distance which is given by $\min\{d_W, d_1(C),d_2\}$ (see Theorem~\ref{thm:window}). Here $d_W$ is the distance of a subsystem spacetime code defined by an appropriate time window around the boundary.
\end{itemize}

This paper is organized as follows. In Section~\ref{sec: prelim}, we review the
background and fix notation. Then in Section~\ref{sec:spacetime-setup} we introduce the masking properties of instantaneous
stabilizer groups and the representation of faults and detectors as spacetime Pauli
operators, illustrated by worked examples. In Section~\ref{sec:setup}, we specify the
circuits the framework applies to. Section~\ref{sec: SS code} constructs the gauge
group of the subsystem spacetime code, determines its center and its decomposition into
gauge qubits, and shows that its distance equals the fault distance of the circuit. In Subsection~\ref{sec:beyond_steady_state}, we discuss the different settings under which the spacetime formalism still holds and in Subsubsection~\ref{subsec: Frameworks for measurements}, we present the different ways of taking into account measurement errors. In
Section~\ref{sec:two-distances}, we relate this distance to the unmasked distances of the
instantaneous stabilizer groups, and Appendix~\ref{sec: counter} gives an example
where the relation fails. Section~\ref{sec: appli} applies the framework to
existing dynamical codes and to the composition of circuits. Finally, we conclude in Section~\ref{sec: conclude} and discuss open problems in Section~\ref{sec: open problems}.

\section{Preliminaries and notation}~\label{sec: prelim}
We begin with preliminary definitions and notations, which includes a brief review of Pauli spaces and subsystem stabilizer codes. 

In this work unless stated otherwise, every Pauli group is taken modulo phases and
regarded as a binary vector space.  Let $\Vcut_n = \mathcal{P}_n/\langle i I\rangle$ be the $n$-qubit Pauli group modulo phases. It is equivalently regarded as a $2n$ dimensional vector space over $\mathbb{F}_2$, equipped with the symplectic form
$\symp:\Vcut_n\times\Vcut_n\to\Ftwo$, where $\symp(E,F)=0$ if $E$ and $F$ commute and $\symp(E,F)=1$ if they anti-commute. Pauli strings still
follow the usual multiplication rules modulo phase. We will use these two representations interchangeably but it will be clear from context.

For a subspace $\mathsf W \subseteq \Vcut_n$, define
\begin{equation}
    \orth{\mathsf W}
    \coloneqq
    \{v\in\Vcut_n:\symp(v,w)=0\text{ for all }w\in\mathsf W\},
    \qquad
    \dual{\mathsf W}
    \coloneqq
    \operatorname{Hom}_{\Ftwo}(\mathsf W,\Ftwo).
\end{equation}
Thus $\orth{\mathsf W}$ is a subspace of $V_n$,
whereas $\dual{\mathsf W}$ is the vector space of linear functionals
on $\mathsf W$.  We call $\mathsf W$ \emph{isotropic} if
$\mathsf W\subseteq\orth{\mathsf W}$ and \emph{nondegenerate} if
$\mathsf W\cap\orth{\mathsf W}=\{0\}$. (Note that this is \emph{not} the same as a nondegenerate code.) 

\begin{definition}[Support and weight]
\label{def:pauli-weight}
For $E\in V_n$, its support and weight are
\begin{equation}
    \operatorname{supp}(E)
    \coloneqq
    \{q\in\{1,\ldots,n\}:E_q\neq I\},
    \qquad
    \operatorname{wt}(E)
    \coloneqq
    |\operatorname{supp}(E)|.
    \label{eq:pauli-weight}
\end{equation}
\end{definition}

\begin{definition}[Stabilizer code]
An $[[n,k,d]]$ stabilizer code, where $n,k,d$ are number of physical qubits, number of logical qubits, and code distance respectively, is a quantum code that is specified by an $(n-k)$-dimensional isotropic subspace $\Stab\subseteq\Vcut_n$, called the stabilizer group. The code space is then given by the joint $+1$-eigenspace of the stabilizer group. \\
\end{definition}

\begin{definition}[Subsystem code]
A $[[n,k,d]]$ subsystem code is defined by a gauge group $\Gspace \subseteq\Vcut_n$ which in general is not isotropic. Its 
stabilizer subspace is given by the center of the gauge group.
\begin{equation}
    \Stab=\Z{\Gspace}
    \coloneqq
    \Gspace\cap\orth{\Gspace}.
\end{equation}
\end{definition}
\begin{definition}[Logical operators]
The bare and dressed logical Pauli subspaces of the subsystem code are
\begin{align}
    \Lbare{\Gspace}
    &\coloneqq
    \orth{\Gspace}\backslash \Z{\Gspace},\\
    \Ldressed{\Gspace}
    &\coloneqq
    \orth{\Z{\Gspace}}\backslash\Gspace.
\end{align}
A bare logical representative commutes with every gauge operator,
whereas a dressed representative need only commute with the stabilizers. \\
\end{definition}

\begin{definition}[Subsystem code distance]
The distance of a subsystem code defined by its gauge group $\Gspace$ is given by
\begin{equation}
    \CodeDistance{\Gspace}
    \coloneqq
    \min\{\operatorname{wt}(E):
    E\in\orth{\Z{\Gspace}}\setminus\Gspace\},
\end{equation}
where $\operatorname{wt}(E)$ is the number of qubits on which $E$
acts non-trivially. Thus, the distance is the minimum weight of a
representative of a non-trivial dressed logical class. For a stabilizer code, this reduces to $\CodeDistance{\Stab}\coloneqq\min\{\operatorname{wt}(E): E \in \orth{\Stab}\backslash \Stab\}$.\\
\end{definition}
\begin{definition} [Destabilizers]
Let $w=(w_1,\ldots,w_r)$ be an ordered basis of an isotropic
subspace $\mathsf W$. An associated set of destabilizers for $w$ is a set of elements
\begin{equation}
    w^\sharp=(w_1^\sharp,\ldots,w_r^\sharp),
    \qquad
    \symp(w_i,w_j^\sharp)=\delta_{ij}.
\end{equation}
We call $w_i^\sharp$ a destabilizer of $w_i$: it anticommutes with
$w_i$ and commutes with every other basis element $w_j$.
Such a choice is in general not unique.

The space
$\mathsf W\oplus\langle w^\sharp\rangle$ is nondegenerate and
provides a symplectic completion of $\mathsf W$.
\end{definition}
\section{Dynamical codes in spacetime}\label{sec:spacetime-setup}
In this section, we introduce constructions and formalisms that will be essential for the rest of the paper. First we introduce the circuit model and propagators. We then introduce the classification of instantaneous stabilizers into unmasked, temporarily masked, and permanently masked stabilizers, which determines the available syndrome information at each time slice. Finally, we define the spacetime fault model and show how detector relations give rise to the stabilizers of the subsystem spacetime code. We illustrate these concepts with explicit examples at the end.
\subsection{Circuit model and propagators}
\label{sec : prelim_circuit_model}
We fix the circuit model, such that the faults, stabilizers
and logical operators are assigned to spacetime locations and the circuit blocks at half-time intervals are
taken to be ideal. We also introduce propagators, which relate a Pauli error at one
cut to an error at the next.

Let $\mathcal{Q}$ denote a quantum circuit composed of Pauli
measurements and Clifford unitaries. A \emph{cut} is an imaginary
time slice through the circuit that separates successive operations
without intersecting any operation. We index the cuts by integer
timesteps $t=0,\ldots,\FinalCut$, where $t=0$ and $t=\FinalCut$
correspond to the input and output boundaries, respectively.
The Pauli space at cut $t$ is denoted $\atcut{\Vcut}{t}$, and the
circuit block between cuts $t$ and $t+1$ is denoted
$\atcut{\CircuitBlock}{t+1/2}$. Each block consists of a family
of Pauli measurements $\atcut{\ms}{t+1/2}=\{m_i^{(t+1/2)}\}_i$ and a collection of Clifford unitaries, whose joint action is denoted
$\atcut{U}{t+1/2}$. All operations within a block have pairwise disjoint supports, so each qubit participates in at most one operation per block. By inserting sufficiently many cuts, any circuit composed of Pauli measurements and Clifford unitaries can be expressed in this form as illustrated in Figure \ref{fig:cicruit model}.

\begin{figure}[ht] 
\centering
\includegraphics[height=4 cm]{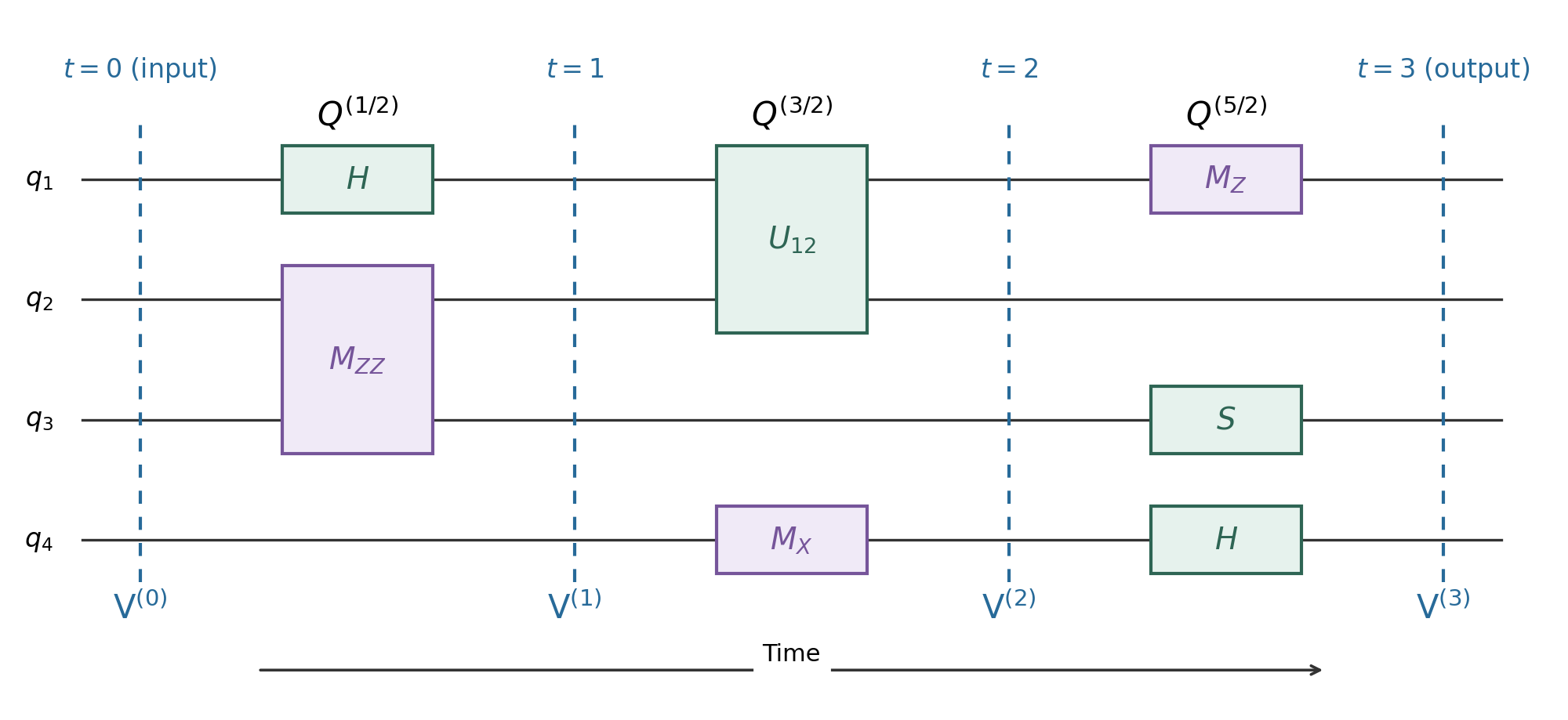}
\caption{The circuit model considered in this paper. We consider imaginary time slices through the circuit that separates successive operations. Each intermediate circuit block consists of unitary operations and measurements having disjoint supports.}
\label{fig:cicruit model}
\end{figure}

\paragraph{Fault model and spacetime Pauli.}
We take the circuit blocks to be ideal, with errors occurring
only at the cuts. In other words, a fault is a Pauli operator acting at cuts. The \emph{spacetime locations} are the pairs $(q,t)$ of a qubit $q$ and a cut $t$ at which $q$ exists. The spacetime Pauli space is
\begin{equation}
    \VST
    \coloneqq
    \bigoplus_{t=0}^{\FinalCut}\atcut{\Vcut}{t},
\end{equation}
with $\atcut{\Vcut}{t}$ denoting the subspace restricted to time $t$. A
spacetime Pauli $E = \bigoplus_t E^{(t)}$ assigns a Pauli $E^{(t)}$ to each
cut, and its weight is the number of spacetime locations on which it acts
nontrivially.  

By a mild abuse of notation, we use $\atcut{E}{t}$ both for an
operator in $\atcut{\Vcut}{t}$ and for its canonical embedding
into $\VST$. Expressions involving operators at distinct cuts
are understood as elements of $\VST$. The symplectic form on
$\VST$ is the direct sum of the forms on the individual slices:
\begin{equation}
    \symp\left(
        \bigoplus_t\atcut{E}{t},
        \bigoplus_t\atcut{F}{t}
    \right)
    =
    \sum_t\symp(\atcut{E}{t},\atcut{F}{t}) \mod{2}.
\end{equation}

Let $E \in \VST$ denote a Pauli fault in spacetime. The spacetime support and weight of a fault configuration are
\begin{align}
    \operatorname{supp}(E)
    &\coloneqq
    \{(q,t):E_q^{(t)}\neq I\},
    \label{eq:spacetime-support}\\
    \operatorname{wt}(E)
    &\coloneqq
    |\operatorname{supp}(E)|
    =
    \sum_{t=0}^{T}
    \operatorname{wt}\!\left(E^{(t)}\right).
    \label{eq:spacetime-weight}
\end{align}
Its spatial support ignores the cut labels:
\begin{equation}
    \operatorname{supp}_{\mathrm{sp}}(E)
    \coloneqq
    \{q:E_q^{(t)}\neq I\text{ for some }t\}.
    \label{eq:spatial-support}
\end{equation}
Thus,
$|\operatorname{supp}_{\mathrm{sp}}(E)|
\leq\operatorname{wt}(E)$.
Unless stated otherwise, the allowed fault space is
$F=V_{\mathrm{ST}}$, equipped with the weight in
Eq.~\eqref{eq:spacetime-weight}.

\paragraph{Dynamical codes} A dynamical code is specified by an initial stabilizer group $\atcut{\Stab}{0} \subseteq \atcut{\Vcut}{0}$ and a sequence of circuit blocks consisting of unitary operators and measurements. In the
measurement-only setting, this sequence is described by a measurement schedule
\begin{equation}
        \mathcal{Q}
        =
        \{
            \atcut{\ms}{1/2},
            \atcut{\ms}{3/2},
            \ldots,
            \atcut{\ms}{\FinalCut-1/2}
        \},
\end{equation}
Measurements in different timeslices need not commute in general. After the block $\atcut{\CircuitBlock}{t-1/2}$, the instantaneous
stabilizer group is updated to $\atcut{\Stab}{t}$ by Clifford
conjugation and the standard stabilizer measurement update rules given in Procedure \ref{lem:stab-update}.
The code space, $\mathcal C^{(t)}$, is the corresponding
simultaneous eigenspace,
\begin{equation}
    \mathcal C^{(t)}
    =
    \{|\psi\rangle:
    s|\psi\rangle=\pm |\psi\rangle
    \text{ for every }s\in\atcut{\Stab}{t}\},
\end{equation}
The eigenvalues are determined by the initial stabilizer eigenvalues and the measurement outcomes.

A Floquet code is a special type of dynamical code with a periodic
measurement schedule, that is
$\atcut{\ms}{t+p+1/2}=\atcut{\ms}{t+1/2}$, where $p$ is the period of the measurement schedule. 

%Stabilizer and logical update rules.
In a dynamical code, every measurement can change the stabilizer group and the logical
representatives. We recall the rules that track both~\cite{fu2024errorcorrectiondynamicalcodes}.
Let $\Stab$ be a stabilizer group with independent generators $s_1,\dots,s_r$, and let the
state lie in the joint eigenspace $s_i\ket{\psi}=(-1)^{a_i}\ket{\psi}$. A Pauli measurement
$m$ has outcome $O(m)\in\{0,1\}$ and leaves the state in the $(-1)^{O(m)}$ eigenspace of $m$.
By the \emph{value} of a logical representative $\ell$ we mean the outcome that a measurement
of $\ell$ would give. \\

\begin{procedure}[Stabilizer update rules~\cite{gottesman1997stabilizer}]
\label{lem:stab-update}

Measure a Pauli operator $m$.
\begin{enumerate}
    \item If $m\in\Stab$, say $m=\pm\prod_{i\in I}s_i$, the outcome is fixed by the signs
    $a_i$, $i\in I$. The state and $\Stab$ are unchanged.
    \item If $m$ anticommutes with some of the generators, choose one of them, $s_j$. Replace
    $s_j$ by $m$, with sign $O(m)$, and replace every other generator $s_i$ that anticommutes
    with $m$ by $s_is_j$, with sign $a_i+a_j$. The outcome is uniformly random, and the rank
    of $\Stab$ is unchanged.
    \item If $m$ commutes with every element of $\Stab$ but $m\notin\Stab$, add $m$ to the
    generators, with sign $O(m)$. The rank of $\Stab$ grows by one.
\end{enumerate}
\end{procedure}

Case~1 is the source of detectors: the outcome is predicted by earlier outcomes and input
signs. Case~2 is the source of masking: the stabilizer $s_j$ leaves the group. For Case~3: If $\Stab$ has
$n-k$ independent generators on $n$ qubits, an operator $m$ as in Case~3 is a nontrivial
logical operator, so Case~3 is a logical measurement.

\begin{procedure}[Logical update rules~\cite{gottesman1997stabilizer}]
\label{lem:logical-update}
Let $\ell$ be a representative of a logical operator of $\Stab$, and measure a Pauli operator
$m$.
\begin{enumerate}
    \item If $m$ commutes with $\ell$, then $\ell$ remains a representative, with the same
    value. If moreover $m=\pm\ell s$ for some $s\in\Stab$, the outcome reveals the value of
    $\ell$.
    \item If $m$ anticommutes with $\ell$ and with some generator, let $s_j$ be the generator
    that Procedure~\ref{lem:stab-update} replaces by $m$. Then $\ell s_j$ commutes with $m$ and
    with every updated generator, so it represents the same logical operator after the
    measurement. Its value is the value of $\ell$ before the measurement, plus the sign
    $a_j$ of $s_j$.
    \item If $m$ anticommutes with $\ell$ and commutes with every element of $\Stab$, then $m$
    is itself a logical operator, and the measurement is a logical measurement. The state
    becomes an eigenstate of $m$, and $\ell$ is replaced by $m$, with sign $O(m)$.
\end{enumerate}
\end{procedure}

The sign $a_j$ in Case~2 is itself a parity of earlier outcomes and input signs, by
Procedure~\ref{lem:stab-update}. A logical value is therefore read off from the final
representative together with the outcomes collected along the way. 

\paragraph{Logical update across a block.}
In the circuit model, a block applies commuting measurements $M^{(t+1/2)}$ on disjoint
supports and a Clifford layer $U^{(t+1/2)}$. We apply Procedure~\ref{lem:logical-update} to the logical representatives $\ell^{(t)}$ at cut $t$, then conjugating by the Clifford layer, to dynamically update the logical representatives and get $\ell^{(t+1)}$ at cut $t+1$. Furthermore, by iterating this process from cut $0$ to cut $T$ for each representative $\ell^{(0)}$, we obtain  $\ell^{(t)}$ for all $t\leq T$ of the same logical operator, together with its ideal measurement outcome at each cut. We also obtain a set of measurements $V_\ell$ which is the set of measurements in the circuit required by the logical update rule to dynamically update the logical representative from $t=0$ to $t=T$. That is, starting from any logical representative, the outcome of the logical representative at $t=T$ depends on the outcomes of the measurements in $V_\ell$. 

\begin{lemma}[$V_\ell$ is well defined up to detectors]
\label{lem:Vell-well-defined}
Let $V_{\ell_1}$ and $V_{\ell_2}$ be logical update sets obtained from two
choices of representatives $\ell_1^{(t)}, \ell_2^{(t)}$ of the same logical
class. Then the measurements appearing in
exactly one of $V_{\ell_1}$ and $V_{\ell_2}$, together with the terminal
readouts of $s_T = \ell_1^{(T)}\ell_2^{(T)}$, form a detector. 
\end{lemma}
\begin{proof}
Write $O(P)$ for the value of a Pauli $P$, i.e.\ the outcome a measurement
of $P$ would give. In every fault-free run, unrolling the logical update
rules gives, for $i=1,2$,
\begin{equation}
  O\!\left(\ell_i^{(T)}\right)
  = O\!\left(\ell_i^{(0)}\right)
  + \sum_{m\in V_{\ell_i}} O(m)
  + O\!\left(s_i^{(0)}\right),
\end{equation}
where $s_i^{(0)}\in \Stab^{(0)}$ collects the initial stabilizers whose signs
enter the bookkeeping. Summing the two relations modulo $2$,
\begin{equation}
  \sum_{m\in V_{\ell_1}} O(m) + \sum_{m\in V_{\ell_2}} O(m)
  = O\!\left(\ell_1^{(0)}\ell_2^{(0)}\right)
  + O\!\left(s_1^{(0)}\right) + O\!\left(s_2^{(0)}\right)
  + O\!\left(\ell_1^{(T)}\ell_2^{(T)}\right).
\end{equation}
On the right-hand side the unknown logical value has cancelled:
$\ell_1^{(0)}\ell_2^{(0)}$ is an initial stabilizer, so the first three
terms are constants fixed by the given input signs. The last
term is the value of $s_T\in \Stab^{(T)}$, which the perfect terminal round
reveals as a parity of recorded outcomes. Moving those readouts to the
left-hand side, the parity of outcomes over $V_{\ell_1}$, $V_{\ell_2}$ and
the terminal readouts of $s_T$ has a fixed value in every fault-free run,
and is therefore a detector.
\end{proof}

\begin{definition}[Propagator space]~\label{def:propagators}
For $P\in\{X,Y,Z\}$, a Pauli $P$ supported on qubit $i$ and timeslice $t$ is denoted as $P_i^{(t)} \in \atcut{\Vcut}{t}$. Define $r_{i}^{(t)}(P)$ as follows:
\begin{equation}
    r_{i}^{(t)}(P)
    \coloneqq
    P_i^{(t)}+ \left[U^{(t+1/2)} P_i \atcut{{U^{\dagger}}}{t+1/2}\right]^{(t+1)},
\end{equation}
where $U^{(t+1/2)}$ is the unitary at the circuit block $t+1/2$. 
The propagator space is
\begin{equation}
    \Rprop
    \coloneqq
    \left\langle
        r_{i}^{(t)}(P)
        :
        i\in[n],\
        P\in\{X,Y,Z\},\
        0\le t<T
    \right\rangle.
\end{equation}
\end{definition}

We classify elementary propagators according to the
commutation relations of $P_i$ with the measurements
between cuts $t$ and $t+1$. \\

\begin{definition}(Commuting Propagator Space)
\label{def:comm_prop_space}
The commuting propagator space $\Rcomm$ is 
defined as the subspace of $\Rprop$ such that any $v\in \mathsf{V}^{(t)}$ that commutes with every measurement in $\atcut{\ms}{t+1/2}$, i.e. $r^{(t)}(v)\coloneqq v^{(t)}+ \left[U^{(t+1/2)} v \atcut{{U^{\dagger}}}{t+1/2}\right]^{(t+1)}$, is in $\Rcomm$.
\begin{equation}
\mathsf{R}_{\mathrm{comm}} := \bigoplus_{t=0}^{T-1} r^{(t)}\big(\langle\mathcal{M}^{(t+1/2)}\rangle^{\perp}\big),
\end{equation}
The remaining propagators are classified as $\Ranti$. Note that $\Ranti$ is not a subspace.
\end{definition}

For a commuting propagator, the error $P_i^{(t)}$ can
be moved through the intervening circuit block without changing any measurement outcome. It then
cancels with $U^{(t+1/2)\dag} P_i^{(t+1)} \atcut{U}{t+1/2}$, so the combined error is
equivalent to the identity.

\subsection{Masking}\label{sec:prelim_masking}
In a dynamical code, not every instantaneous stabilizer helps with error correction:
an anticommuting measurement can destroy a stabilizer before its value is ever
computed. Masking, introduced in Ref.~\cite{fu2024errorcorrectiondynamicalcodes},
makes this precise and determines which stabilizers at each cut carry syndrome
information.

Fix a cut $t$ and a measurement window $W$ till some cut $T$. The elements of the instantaneous stabilizer group $\atcut{\Stab}{t}$ are classified according to whether their syndromes can be reconstructed from the measurement outcomes within this window.\\

\begin{definition}[Unmasked stabilizers]
Unmasked stabilizers are stabilizers whose syndromes can be expressed as a parity of outcomes from the specified measurement
sequence. The unmasked stabilizers form a subspace $\atcut{\Uspace}{t} \subseteq \atcut{\Stab}{t}$,
which can be determined using the distance algorithm of
\cite{fu2024errorcorrectiondynamicalcodes}. Its eigenvalue, in the absence of errors, is revealed by a parity of outcomes in $W$.\\
\end{definition}

\begin{definition}[Permanently masked stabilizers]

A permanently masked stabilizer is a stabilizer whose syndrome information has been irreversibly lost during the measurement sequence. The measurement record and the final state at the end of the sequence together provide no information about the syndrome of a permanently masked stabilizer. Thus, even extending the measurement sequence with new measurements cannot recover its syndrome.\\
\end{definition}

\begin{definition}[Temporarily masked stabilizers]

A temporarily masked stabilizer is a stabilizer whose syndrome cannot be inferred from the given measurement window. However, the final state at the end of the window retains full information about its syndrome, so it is possible to infer the syndrome by adding suitable future measurements. 
\end{definition}

Choose a basis of $\atcut{\Stab}{t}$ adapted to this classification,
and let $\atcut{\Tspace}{t}$ and $\atcut{\Pspace}{t}$ denote the set of
the selected temporarily and permanently masked generators,
respectively. Then
\begin{equation}
    \atcut{\Stab}{t}
    =
    \langle\atcut{\Uspace}{t}
    \cup
    \atcut{\Tspace}{t}
    \cup
    \atcut{\Pspace}{t}\rangle.
\end{equation}
The unmasked subspace is determined by the measurement window, whereas the displayed decomposition also involves choices of
complementary subspaces and generators.

For elements $\Telt{b}{t}$ of $\atcut{\Tspace}{t}$
and $\Pelt{c}{t}$ of $\atcut{\Pspace}{t}$ in a chosen basis of $\atcut{\Stab}{t}$, we can choose corresponding $\Tdual{b}{t}$ and $\Pdual{c}{t}$ as their destabilizers. The destabilizers for permanently masked stabilizers are fixed by the measurement scheme (Algorithm 3 in \cite{fu2024errorcorrectiondynamicalcodes}) whereas there is a freedom in choosing the destabilizers of temporarily masked stabilizers. 
 This freedom is in general fixed by maximizing the unmasked distance (Definition \ref{def:unmasked_distance}). The destabilizers of the temporarily masked stabilizers and the permanently masked stabilizers form the subsets, $\boldsymbol{\tau}^{\sharp(t)}$ and $\boldsymbol{p}^{\sharp(t)}$ respectively. Thus each destabilizer anti-commutes with its associated generator and commutes with all other generators in that basis. 

We briefly explain how the destabilizers for permanently masked stabilizers are fixed below.

\paragraph{Permanently masked stabilizers and their destabilizers.}
In the distance algorithm~\cite{fu2024errorcorrectiondynamicalcodes}, a stabilizer at time $t$ becomes permanently masked when a measurement at $t+1/2$ anti-commutes with its current representative in $C$ and the algorithm removes that representative from $C$ (Case~2 in the Distance Algorithm). From then on, its eigenvalue is stored neither in the state nor in any parity of recorded outcomes, so no later measurement can recover it.

Its destabilizer is fixed by the circuit through the time reversal of this process, and is computed by Algorithm~3 of~\cite{fu2024errorcorrectiondynamicalcodes}. The masking measurement $m$ anticommutes with the representative of the permanently masked stabilizer $p(t)$ at time $t$. Then $(p(t),m(t)\coloneqq m)$ forms a gauge qubit at time $t$. By treating this pair as a logical qubit, we can apply logical update rule backward in time to find a complete set of permanently masked stabilizers and destabilizers at previous timesteps that stem from the pair $(p(t),m(t))$. The destabilizer is therefore fixed by masking measurements and the logical update rules. See Example~\ref{ex: masking} for an example of computing the destabilizers.

Temporarily masked stabilizers have not been masked. Thus there is a degree of freedom in choosing its masking measurement, which we can optimize by say maximizing the unmasked distance. 

\begin{definition}[Instantaneous gauge group] \label{def : IGG}
Not all the syndrome information in $\atcut{\Stab}{t}$ can be learned from the measurement window. Only the syndromes in $\atcut{\Uspace}{t}$ are available for distinguishing errors. Treating the temporarily and permanently masked stabilizers as gauge degrees of freedom gives the gauge group at cut $t$:
\begin{equation}
    \atcut{\Gspace}{t}
    \coloneqq
    \left\langle
        \atcut{\Uspace}{t},
        \atcut{\Tspace}{t},
        \atcut{\Pspace}{t},
        \boldsymbol{\tau}^{\sharp(t)},
        \boldsymbol{p}^{\sharp(t)}
    \right\rangle.
\end{equation}
Its stabilizer center is $\Z{\atcut{\Gspace}{t}} := \atcut{\Uspace}{t}$. \\
\end{definition}

\begin{definition}[Unmasked distance]
\label{def:unmasked_distance}
For the masked stabilizer code at cut $t$, with the specified measurement window and a chosen gauge group $\atcut{\Gspace}{t}$,
the unmasked distance is
\begin{equation}
    d_{\mathrm{u}}(\atcut{\Gspace}{t})
    \coloneqq
    \min\left\{
        \operatorname{wt}(E):
        E\in\left(\atcut{\Uspace}{t}\right)^\perp \setminus \atcut{\Gspace}{t}
    \right\}.
\end{equation}
\end{definition}
The gauge group depends on the choice of destabilizers for the temporarily masked stabilizers. The unmasked distance can therefore be maximized over choices of destabilizers for $\boldsymbol{\tau}^{\sharp(t)}$.

The distance algorithm of \cite{fu2024errorcorrectiondynamicalcodes} determines the unmasked
syndromes and classifies the stabilizer
generators. Starting from an initial stabilizer group
$\atcut{\Stab}{0}$, it maintains two sets: $C$, which tracks how
the initial stabilizer generators evolve, and $V$, which tracks
how the recorded measurements can be combined to determine
their syndromes. Initially, $C$ consists of generators of
$\atcut{\Stab}{0}$ and $V=\varnothing$.

The algorithm processes the measurement blocks sequentially.
At the end of the window, $\langle C\rangle\cap\langle V\rangle$
identifies the evolved representatives of the unmasked initial
stabilizers. Tracking these representatives back to
$\atcut{\Stab}{0}$ gives the unmasked subspace, together with the
measurement-outcome parities that determine its syndromes. The update rules are given in \cite{fu2024errorcorrectiondynamicalcodes}. Iterating through each ISG gives $\Gspace^{(t)}$. The information from unmasking stabilizers can also be used to obtain the detectors. A worked example is
presented in Subsection~\ref{sec : prelim example}.

\subsection{Detectors, and spacetime Pauli representations} \label{sec : prelim_fauls and detectors}

\paragraph{From detector relations to detector stabilizers.}
A detector $\delta$ is a parity relation among measurement
outcomes whose value is fixed in the absence of faults \cite{delfosse2023spacetimecodescliffordcircuits}.
Let $M_\delta$ denote the set of measurement events entering
this relation. Taking the fault-free parity to be zero,
\begin{equation}
    \delta \::\:
    \sum_{m_i^{(r+1/2)}\in M_\delta}
    O\!\left(m_i^{(r+1/2)}\right)
    =0
    \pmod{2}.
\end{equation}

Each detector determines a subset of unmasked stabilizers
$\Uelt{\delta}{t}$ for some subset of timeslices. To construct $\Uelt{\delta}{t}$, for a fixed $t$, define
the product of measurements that are involved in $\delta$ from the circuit block $\atcut{Q}{t+1/2}$ as
\begin{equation}
    d_\delta^{(t+1/2)}
    \coloneqq
    \sum_{\substack{
        i:\,m_i^{(t+1/2)}\in M_\delta
    }}
    m_i^{(t+1/2)}.
\end{equation}
Let
\begin{equation}
    C_{r,t}
    \coloneqq
    \atcut{U}{r-1/2}\cdots\atcut{U}{t+1/2},
    \qquad C_{t,t}\coloneqq I,
\end{equation}
denote the ordered product of Clifford layers between cuts
$t$ and $r$, where $r \geq t$. We assume a perfect round of measurements of stabilizers in $\Stab^{(T)}$ at $T+1/2$.
The unmasked stabilizer associated with detector $\delta$
at cut $t$ is
\begin{equation}\label{eq:detector-cut-stabilizer}
    \Uelt{\delta}{t}
    =
    \sum_{r=t}^{\FinalCut}
    C_{r,t}^{\dagger}
    d_\delta^{(r+1/2)}
    C_{r,t}.
\end{equation}
 In the absence of Clifford gates, this reduces to
\begin{equation}
    \Uelt{\delta}{t}
    =
    \sum_{r=t}^{\FinalCut}d_\delta^{(r+1/2)}.
\end{equation}
The detector relation makes $\Uelt{\delta}{t}$ an unmasked
stabilizer: its eigenvalue can be inferred from the
corresponding future measurement outcomes. 

Collecting these operators across all cuts defines the
detector-to-spacetime map. Let
   $ \Sdet$
be the span of the detector stabilizers. Section~\ref{sec: SS code} shows that these are exactly the stabilizers of the subsystem spacetime code
\begin{equation}
    \DetStab{\delta}
    \coloneqq
    \DetMap(\delta)
    \coloneqq
    \bigoplus_{t=0}^{\FinalCut}\Uelt{\delta}{t},
    \qquad
    \Sdet\coloneqq\operatorname{im}\DetMap.
\end{equation}
For a spacetime error
$E=\bigoplus_{t=0}^{\FinalCut}E^{(t)}$, its syndrome from the detector $\delta$ is
\begin{equation}\label{eq:syndrome}
    \sigma_\delta(E)
    =\symp(E,s_\delta)
    =\sum_{t=0}^{\FinalCut}
      \symp\!\left(E^{(t)},\Uelt{\delta}{t}\right)
    \pmod{2}.
\end{equation}

\begin{remark}
$\Sdet$ can be constructed ISG by ISG using the distance algorithm.
For each ISG, the algorithm identifies an unmasked stabilizer together
with a set of measurement outcomes whose product reveals its
syndrome.  These data determine a detector
$\DetElt\in\DetectorSpace$ and its corresponding stabilizer
$\DetStab{\DetElt}=\DetMap(\DetElt)$ in the spacetime code. Conversely, restricting the spacetime representation of a detector to each cut gives its associated
unmasked stabilizer.
\end{remark}

\subsection{Worked out examples}\label{sec : prelim example}
We illustrate the definitions above on small examples, computing the unmasked and
permanently masked stabilizers, the unmasked distance, and the detector stabilizers
explicitly. 
	
\begin{example}
Suppose the Hastings and Haah's Floquet code~\cite{Hastings2021dynamically} is initialized with $XX$ checks at round 0, followed by $YY$ checks at round 1 and $ZZ$ checks at round 2. Figure~\ref{fig:hhfloquet} shows examples of spacetime code stabilizers with weight at most $24$. 

\end{example}

\begin{figure}[ht]
\centering
\includegraphics[width = 0.28\textwidth]{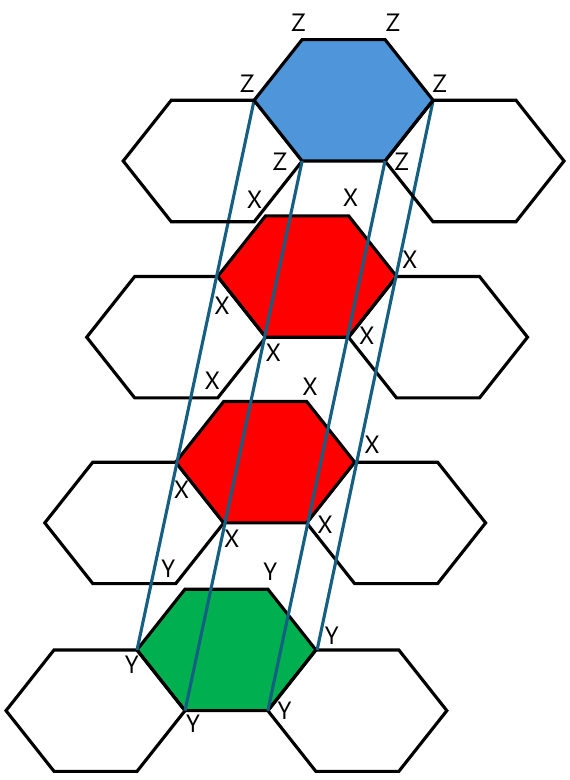}
\includegraphics[width=0.4\textwidth]{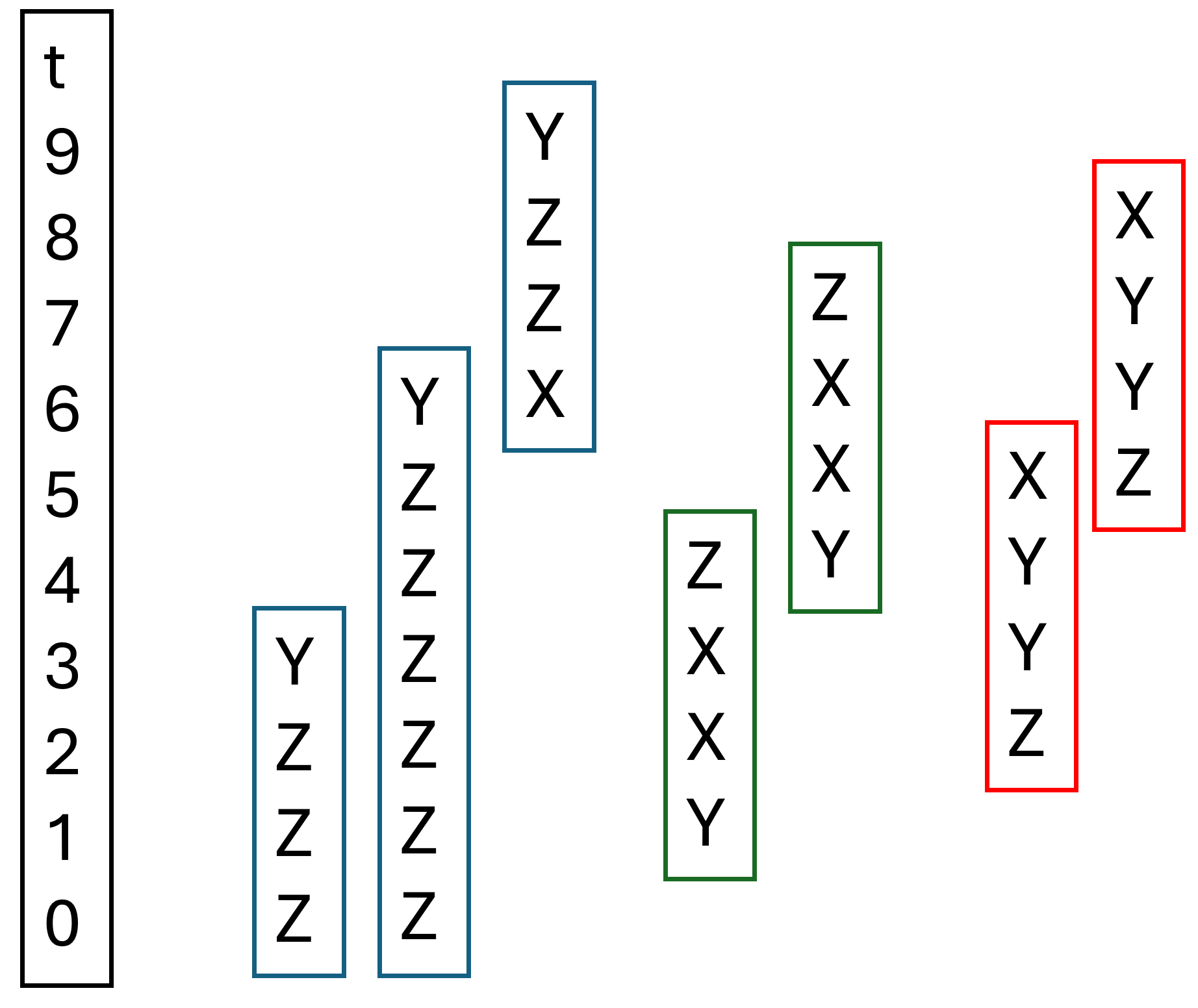}
\caption{Spacetime code structure for the Hastings and Haah's Floquet code. (Left) A stabilizer in the spacetime code with weight 24. (Right) Each entry $X,Y,Z$ denotes a weight-6 Pauli on a hexagon of the same Pauli type in the honeycomb lattice. Operators on the same column are supported by the same hexagon but at different timeslices.}
\label{fig:hhfloquet}
\end{figure}
\FloatBarrier
\begin{example}~\label{ex: masking} 
We work through the following example to illustrate all of the above concepts. Specifically, we use the distance algorithm in \cite{fu2024errorcorrectiondynamicalcodes} to calculate the \hyperlink{unmasked stabilizers}{unmasked stabilizers}, \hyperlink{permanently masked stabilizers}{permanently masked stabilizers}, \hyperlink{unmasked distance}{unmasked distance}, \hyperlink{spacetime stabilizer}{spacetime code's stabilizers}. For brevity, we will only calculate for the first ISG of the measurement sequence. 

\begin{figure}[H]
\centering
\resizebox{\linewidth}{!}{%
\begin{tikzpicture}[
    font=\small,
    text=algink,
    >={Latex[length=2.2mm,width=1.4mm]},
    state/.style={
        draw=alglinegray,
        fill=white,
        rounded corners=2.2pt,
        minimum height=10mm,
        text width=72mm,
        align=left,
        inner xsep=3.2mm,
        inner ysep=2.1mm
    },
    vstate/.style={
        state,
        text width=63mm
    },
    initial/.style={
        fill=algslatepale
    },
    final/.style={
        draw=algblue,
        very thick,
        fill=algbluepale
    },
    event/.style={
        draw=algnavy,
        fill=algnavy,
        text=white,
        rounded corners=5pt,
        minimum width=29mm,
        minimum height=8.5mm,
        font=\sffamily\bfseries\footnotesize,
        align=center
    },
    casetag/.style={
        fill=algbluepale,
        text=algnavy,
        rounded corners=3pt,
        font=\sffamily\bfseries\scriptsize,
        inner xsep=2.2mm,
        inner ysep=1.1mm
    },
    flow/.style={
        draw=alglinegray,
        line width=.8pt,
        ->
    },
    changelabel/.style={
        fill=white,
        text=algmuted,
        font=\sffamily\scriptsize,
        inner xsep=1.5mm,
        inner ysep=.4mm
    }
]

% ============================================================
% Column headings
% ============================================================

\node[
    anchor=west,
    font=\sffamily\bfseries\large,
    text=algnavy
] at (-3.85,1.30)
    {Current representatives $C_j$};

\node[
    anchor=center,
    font=\sffamily\bfseries\large,
    text=algnavy
] at (5.15,1.30)
    {Measurement};

\node[
    anchor=east,
    font=\sffamily\bfseries\large,
    text=algnavy
] at (13.90,1.30)
    {Recorded operators $V_j$};

\draw[algnavy,line width=1.2pt]
    (-3.85,.86)--(13.90,.86);

% ============================================================
% Initial state
% ============================================================

\node[state,initial] (c0) at (0,0)
{
    $C_0$
    \hfill
    \textcolor{algmuted}{initial}
    \par\smallskip
    $s_1\mapsto\texttt{YXXXYY}
     \qquad
     s_2\mapsto\texttt{ZZIYXI}$
    \par
    $s_3\mapsto\texttt{ZZZXYZ}
     \qquad
     s_4\mapsto\texttt{XIZZXY}$
};

\node[vstate,initial] (v0) at (10.55,0)
{
    $V_0$
    \hfill
    \textcolor{algmuted}{initial}
    \par\smallskip
    $\varnothing$
};

% ============================================================
% State after m_1
% ============================================================

\node[state,draw=algamber] (c1) at (0,-2.15)
{
    $C_1$
    \par\smallskip
    $s_1\mapsto\texttt{YXXXYY}
     \qquad
     s_2\mapsto\texttt{ZZIYXI}$
    \par
    $s_3\mapsto\texttt{ZZZXYZ}$
};

\node[vstate] (v1) at (10.55,-2.15)
{
    $V_1$
    \par\smallskip
    $\{\texttt{IIXZYZ}\}$
};

% ============================================================
% State after m_2
% ============================================================

\node[state] (c2) at (0,-4.30)
{
    $C_2$
    \par\smallskip
    $s_1\mapsto\texttt{YXXXYY}
     \qquad
     s_2\mapsto\texttt{ZZIYXI}$
    \par
    $s_3\mapsto\texttt{ZZZXYZ}$
};

\node[vstate] (v2) at (10.55,-4.30)
{
    $V_2$
    \par\smallskip
    $\{\texttt{IIXZYZ},\texttt{XYZXYZ}\}$
};

% ============================================================
% State after m_3
% ============================================================

\node[state,draw=algblue] (c3) at (0,-6.45)
{
    $C_3$
    \par\smallskip
    $s_1\mapsto\texttt{YXXXYY}
     \qquad
     s_2\mapsto\texttt{ZZIYXI}$
    \par
    $\textcolor{algblue}{
        \mathbf{s_3\mapsto\texttt{ZZYYII}}
     }$
};

\node[vstate,draw=algblue] (v3) at (10.55,-6.45)
{
    $V_3$
    \par\smallskip
    $\{
        \textcolor{algblue}{\mathbf{\texttt{XYYYII}}},
        \texttt{IIIIIY}
      \}$
};

% ============================================================
% State after m_4
% ============================================================

\node[state,draw=algamber] (c4) at (0,-8.60)
{
    $C_4$
    \par\smallskip
    $s_1\mapsto\texttt{YXXXYY}
     \qquad
     s_3\mapsto\texttt{ZZYYII}$
};

\node[vstate] (v4) at (10.55,-8.60)
{
    $V_4$
    \par\smallskip
    $\{
        \texttt{XYYYII},
        \texttt{IIIIIY},
        \texttt{IIIIYI}
      \}$
};

% ============================================================
% State after m_5
% ============================================================

\node[state] (c5) at (0,-10.75)
{
    $C_5$
    \par\smallskip
    $s_1\mapsto\texttt{YXXXYY}
     \qquad
     s_3\mapsto\texttt{ZZYYII}$
};

\node[vstate,draw=algblue] (v5) at (10.55,-10.75)
{
    $V_5$
    \par\smallskip
    $\{
        \texttt{XYYYII},
        \textcolor{algblue}{\mathbf{\texttt{IIIIYY}}},
        \texttt{IIIIXX}
      \}$
};

% ============================================================
% Final state after m_6
% ============================================================

\node[state,final] (c6) at (0,-12.90)
{
    $C_6$
    \hfill
    \textcolor{algblue}{final}
    \par\smallskip
    $s_1\mapsto\texttt{YXXXYY}
     \qquad
     s_3\mapsto\texttt{ZZYYII}$
};

\node[vstate,final] (v6) at (10.55,-12.90)
{
    $V_6$
    \hfill
    \textcolor{algblue}{final}
    \par\smallskip
    $\{
        \texttt{XYYYII},
        \texttt{IIIIYY},
        \texttt{IIIIXX},
        \texttt{ZZYYYY}
      \}$
};

% ============================================================
% Measurements and case labels
% ============================================================

\foreach \i/\yy/\op/\case in {
    1/-1.075/IIXZYZ/2,
    2/-3.225/XYZXYZ/1,
    3/-5.375/IIIIIY/4,
    4/-7.525/IIIIYI/2,
    5/-9.675/IIIIXX/3,
    6/-11.825/ZZYYYY/1
}{
    \node[event] (m\i) at (5.15,\yy)
        {$m_{\i}=\texttt{\op}$};

    \node[
        casetag,
        anchor=west
    ] at ([xshift=1.4mm]m\i.east)
        {Case \case};
}

% ============================================================
% Evolution of C
% ============================================================

\draw[flow,draw=algamber]
    (c0.south)
    --node[changelabel,text=algamber]
        {remove $s_4$}
    (c1.north);

\draw[flow]
    (c1.south)
    --node[changelabel]
        {unchanged}
    (c2.north);

\draw[flow,draw=algblue]
    (c2.south)
    --node[changelabel,text=algblue]
        {$s_3\leftarrow s_3 \cdot m_1$}
    (c3.north);

\draw[flow,draw=algamber]
    (c3.south)
    --node[changelabel,text=algamber]
        {remove $s_2$}
    (c4.north);

\draw[flow]
    (c4.south)
    --node[changelabel]
        {unchanged}
    (c5.north);

\draw[flow]
    (c5.south)
    --node[changelabel]
        {unchanged}
    (c6.north);

% ============================================================
% Evolution of V
% ============================================================

\draw[flow]
    (v0.south)
    --node[changelabel]
        {add $m_1$}
    (v1.north);

\draw[flow]
    (v1.south)
    --node[changelabel]
        {add $m_2$}
    (v2.north);

\draw[flow,draw=algblue]
    (v2.south)
    --node[changelabel,text=algblue]
        {$m_1 \cdot m_2=\texttt{XYYYII}$}
    (v3.north);

\draw[flow]
    (v3.south)
    --node[changelabel]
        {add $m_4$}
    (v4.north);

\draw[flow,draw=algblue]
    (v4.south)
    --node[changelabel,text=algblue]
        {$m_3 \cdot m_4=\texttt{IIIIYY}$}
    (v5.north);

\draw[flow]
    (v5.south)
    --node[changelabel]
        {add $m_6$}
    (v6.north);

% ============================================================
% Dotted guides connecting measurements to updates
% ============================================================

\foreach \i in {1,...,6}{
    \draw[densely dotted,draw=alglinegray]
        (m\i.west)--++(-1.05,0);

    \draw[densely dotted,draw=alglinegray]
        ([xshift=10mm]m\i.east)--++(1.15,0);
}

% ============================================================
% Legend
% ============================================================

\node[
    anchor=north west,
    font=\sffamily\footnotesize,
    text=algmuted,
    text width=176mm,
    align=left
] at (-3.85,-13.68)
{
    Each box gives the state immediately after the indicated
    measurement.
    \textcolor{algamber}{Amber} marks permanently masked stabilizers;
    \textcolor{algblue}{blue} marks a nontrivial change of basis.
    Pauli phases are suppressed.
};

\end{tikzpicture}%
}

\caption{
Evolution of the distance-algorithm data.
The left and right tracks show the states of \(C_j\) and \(V_j\),
respectively, after each measurement.
The central nodes identify the measurement and the corresponding
update rule.
}
\label{fig:distance-algorithm-six-qubit}
\end{figure}
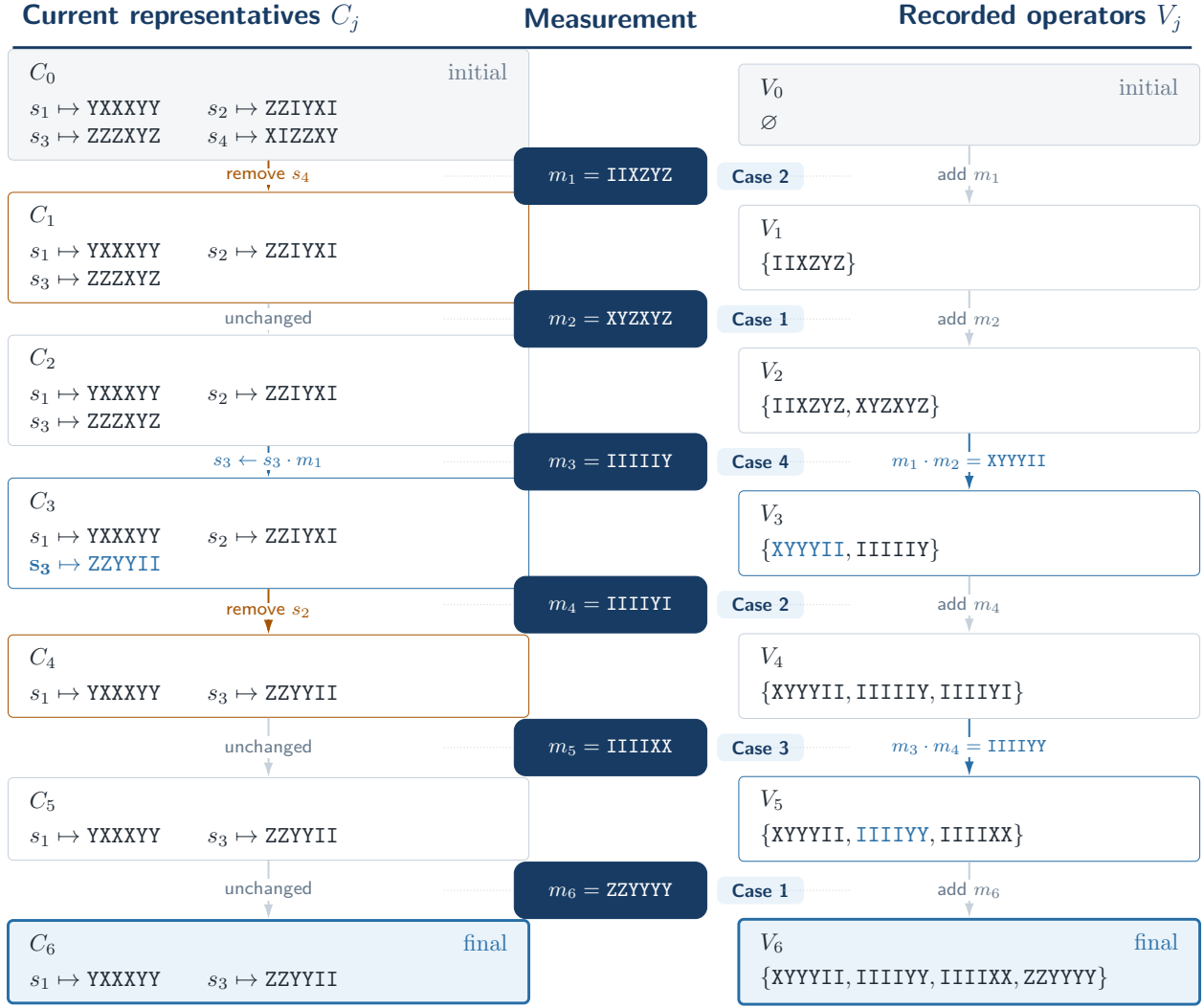

\paragraph{Unmasked stabilizer and detector cell.}
The \hypertarget{unmasked stabilizers}{unmasked stabilizers} can be obtained from the distance algorithm by computing generators of $\langle C \rangle \cap \langle V \rangle$. Operationally, an initial stabilizer is unmasked when the recorded measurements in $V$ determine its syndrome.  At
the end of the sequence,
\begin{equation}
    \texttt{ZZYYII}
      = \texttt{IIIIYY}\cdot \texttt{ZZYYYY}
      \in \langle C_6 \rangle\cap\langle V_6 \rangle .
\end{equation}
Since $\texttt{ZZYYII}$ is the current representative associated with the initial stabilizer $s_3 = \texttt{ZZZXYZ}$, the latter is unmasked. Its syndrome can be obtained from the measurement sequence $\{m_1, m_3, m_4, m_6\}$ as,
\begin{equation}
    O(s_3) = O(\texttt{IIXZYZ})+O(\texttt{IIIIIY})+O(\texttt{IIIIYI})+O(\texttt{ZZYYYY})
\end{equation}
    
\paragraph{Permanently masked stabilizers and destabilizers.}  The \hypertarget{permanently masked stabilizers}{permanently masked stabilizers} are given by the stabilizers that are removed from $C$ during computation. In this example the permanently masked stabilizers are given by, 
\begin{equation}
    p_4 = s_4=\texttt{XIZZXY},
    \qquad
    p _2 = s_2=\texttt{ZZIYXI}.
\end{equation}
The destabilizer for a permanently masked stabilizer is computed from
the round that masks it and then propagated backward to the initial
cut.  In this example,
\begin{equation}
    p_4^\sharp = \texttt{IIXZYZ},
    \qquad
    p_2^\sharp
    =\texttt{IIXZIZ}
    =\texttt{IIIIYI}\cdot\texttt{IIXZYZ}.
\end{equation}

\paragraph{Unmasked distance.}
Suppose that the temporarily masked generator $s_1$ is \hypertarget{unmasked distance}{unmasked} in a later
measurement round.  The gauge group associated with the initial ISG is then
\begin{equation}
    \Gspace
      =\langle s_1,s_2,s_3,s_4,p_4^\sharp,p_2^\sharp\rangle,
    \qquad
    \Z{\Gspace}=\langle s_1,s_3\rangle.
\end{equation}
Consequently,
\begin{equation}
    d_{\mathrm u}
      =\min\bigl\{\operatorname{wt}(P):P\in\Z{\Gspace}^{\perp}\setminus\Gspace\bigr\}.
\end{equation}

\paragraph{Spacetime-code stabilizer.}
The \hypertarget{spacetime stabilizer}{spacetime-code stabilizer} is obtained from the cumulative
partial products of the corresponding detector relation. For
$s_3=\texttt{ZZZXYZ}$, the detector relation is
\begin{equation}
    s_3=m_1m_3m_4m_6.
\end{equation}
Consequently, the corresponding spacetime-code stabilizer is
\begin{equation}
\begin{split}
    \widetilde{s}_3={}&
      (\texttt{ZZZXYZ})^{(0)}
      +(\texttt{ZZYYII})^{(1)}
      +(\texttt{ZZYYII})^{(2)}\\
      &+
      (\texttt{ZZYYIY})^{(3)}
      +(\texttt{ZZYYYY})^{(4)}
      +(\texttt{ZZYYYY})^{(5)}.
\end{split}
\end{equation}

 \end{example}

 \paragraph{Logical error} 
    
    A logical operator can be dynamically updated by a sequence of measurements and its logical outcome is also updated by the measurement outcomes. Thus, if there is a measurement error but the measurement is used to update the logical, the logical outcome will be updated incorrectly. 
    
    For example, a code is initialized in the +1 eigenstate of the logical operator $\texttt{XXXX}$. Suppose we measure $\texttt{IZZI}$ but obtained a measurement error of $-1$ instead of $+1$. In the next time step, when we measure $\texttt{IIYI}$, the logical representative is updated to $\texttt{XYYX}$ which should have an outcome of $+1$, but we would have mistaken it to have the opposite sign. In a dynamical code, a single measurement error can therefore cause a logical error, so the spacetime location of measurement errors matters more than in a static code.

\section{Setup}\label{sec:setup}
The subsystem spacetime code depends on more than the circuit itself: it depends on
where faults may occur, how measurement errors are represented, and what is assumed
at the boundaries of the window. This section describes the choices we make for these ingredients. We first describe the basic construction: an input set of qubits with input instantaneous stabilizer generators, a quantum circuit consisting of 
Pauli measurements and Clifford unitaries, and a terminal round of perfect stabilizer measurements
(Section~\ref{sec:setup-basic}). We then describe how we assign fault locations to measurements (Section~\ref{sec:setup-frameworks}).

\subsection{Basic circuit model and terminal boundary}
\label{sec:setup-basic}

\paragraph{Circuit blocks and cuts.}
We use the circuit model of Subsection~\ref{sec : prelim_circuit_model} and fix a finite window. The window
is divided into \emph{cuts} $t = 0, 1, \dots, T$ and \emph{blocks}
$\atcut{\CircuitBlock}{t+1/2}$ between consecutive cuts. Each block contains
Pauli measurements and Clifford unitaries acting on pairwise disjoint
supports, so each qubit takes part in at most one elementary operation per
block. The input stabilizer group $\mathcal{S}^{(0)}$, including the signs of its
generators, and the input logical operators are given. \\

\begin{remark}
In our constructions, we always use a window with an odd number of cuts,
appending one idle block before the terminal round when necessary. For an even number of cuts, the distance does not change but we cannot include the products of adjacent logical representatives (Section \ref{sec:adjacent-logical-generators}) because otherwise we will have $k=0$.

\paragraph{Terminal boundary.}
Unless stated otherwise, we append a perfect terminal readout at $T+\tfrac12$, immediately after the final fault cut $T$. This readout measures a generating set of $\Stab^{(T)}$. The stabilizers in the final timestep are then unmasked, i.e. $\atcut{\Uspace}{\FinalCut} = \atcut{\Stab}{\FinalCut}$. This terminal round prevents faults near the end of the window from escaping detection merely because the measurement record
ends too soon. For example, consider a data-qubit Pauli error immediately before the last noisy measurement round. If it
anticommutes with just one measured stabilizer generator, a measurement error that flips that generator's reported outcome can conceal the data error. With no subsequent measurements,
this two-fault event can remain undetectable. A final perfect measurement of that generator reveals the residual syndrome. We therefore regard the perfect terminal round as an \emph{artificial boundary condition} that is required to have a well-defined fault distance. 

\end{remark}

\subsection{Representing measurement errors}
\label{sec:setup-frameworks}

In a model where faults live only on the cuts, a flipped measurement outcome has no location of its own. Instead we represent measurement errors as Pauli errors on adjacent cuts (See figure~\ref{fig:formalism_a_demo}). 

\begin{definition}
Let $m_i^{t+1/2}$ be a measurement between the time slices $t$ and $t+1$. A flipped measurement outcome is represented as any propagator of the form $P^{(t)}P^{(t+1)}$ that also satisfies the following property: 
\begin{equation}
\operatorname{supp}_{\mathrm{sp}} (P)\subseteq \operatorname{supp}(m_i) \qquad \text{and} \qquad \symp(P,m_i)=1
\end{equation}
where, $\operatorname{supp}_{\mathrm{sp}}\!\left(P\right)$ denotes the \emph{spatial} support of $P$ (Equation~\ref{eq:spatial-support}). Thus, a measurement error has non-unique Pauli representations in $\VST$.
\end{definition}

Faults in
the terminal round cannot be represented, which is consistent with that
round being perfect. For this framework a weight one measurement error is treated as at least a weight two spacetime Pauli error. (If a block were allowed to contain overlapping commuting measurements, minimum weight of $P$ could be higher). 

\begin{figure}[ht]
    \centering
    \includegraphics[width=\linewidth]{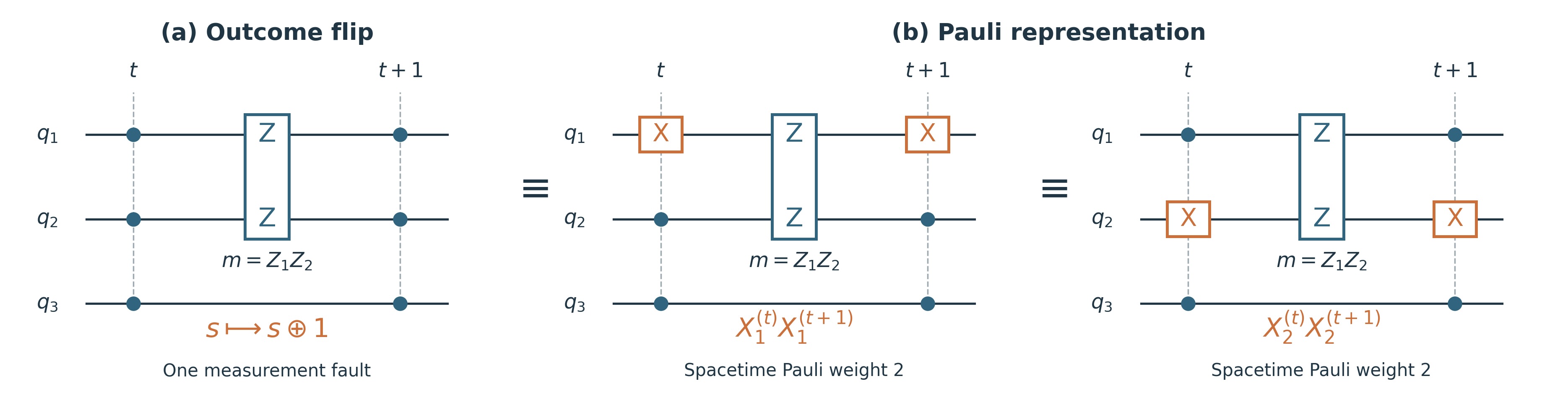}
    \caption{Representing a measurement error on adjacent cuts. (a) The reported outcome of measuring $m=Z_1Z_2$ is flipped. (b) Two equivalent representations are $X_1^{(t)}X_1^{(t+1)}$ and $X_2^{(t)}X_2^{(t+1)}$: the Pauli before the measurement flips its outcome, while
    the Pauli afterward cancels the residual data error. Thus, a single measurement fault admits non-unique representations of spacetime Pauli weight two.}
    \label{fig:formalism_a_demo}
\end{figure}

\subsection{Steady state assumption}\label{sec:steady_stage}
In Section~\ref{sec: SS code}, we construct the subsystem spacetime code under a \emph{steady-state} assumption \cite{blackwell2025code}. This assumption is not necessary and is only used for clearer explanation and we will show in Section \ref{sec:beyond_steady_state} that the subsystem spacetime can be extended to beyond the steady state to incorporate a general Clifford circuit. In the steady state assumption, the measurement sequence preserves the encoded logical information throughout the circuit. In particular, measurements may update the instantaneous stabilizer group, but they must not measure a nontrivial logical operator and thereby project the encoded state onto one of its eigenspaces. This assumption allows us to track a fixed logical Pauli space across the measurement window, even though its physical representatives may change from one cut to the next.

\section{The Subsystem Spacetime Code}
\label{sec: SS code}
We first provide a simple recipe for explicitly constructing the subsystem spacetime code of any given Clifford circuit under the assumptions of Section \ref{sec:setup} in Subsection~\ref{sec:spacetime-gauge-construction}. Then, we examine the elements of its gauge group and provide a natural classification in a way that resembles the masking properties of instantaneous stabilizer groups. We also demonstrate that the construction possesses properties that make it useful for verifying a circuit's fault tolerance, that is the distance from the corresponding subsystem spacetime code is exactly the fault distance of the circuit. Lastly we consider  extensions beyond the steady state assumption in Subsection \ref{sec:beyond_steady_state}.

\subsection{Construction of the gauge group}
\label{sec:spacetime-gauge-construction}
In this subsection, we provide a simple recipe useful for explicitly constructing the subsystem spacetime code's gauge group.

At each cut $t$, let $\atcut{\Gspace}{t}$ be the
instantaneous gauge group of Definition~\ref{def : IGG},
determined by the specified measurement window and the chosen
destabilizers of masked stabilizers. We regard
$\atcut{\Gspace}{t}$ as a subspace of $\VST$ supported
only at cut $t$. Let $\Rcomm$ be the commuting propagator
space defined in Definition~\ref{def:comm_prop_space}. \\

\begin{definition}[Subsystem spacetime gauge group]
\label{def:spacetime-gauge-group}
The gauge group of the subsystem spacetime code is
\begin{equation}
    \Gspace
    \coloneqq
    \left\langle
        \atcut{\Gspace}{0},
        \atcut{\Gspace}{1},
        \ldots,
        \atcut{\Gspace}{\FinalCut},
        \Rcomm
    \right\rangle
    =
    \left(
        \bigoplus_{t=0}^{\FinalCut}
        \atcut{\Gspace}{t}
    \right)
    \oplus \Rcomm.
    \label{eq:spacetime-gauge-group}
\end{equation}
The subsystem spacetime code is the subsystem code on
$\VST$ defined by this gauge group.
\end{definition}

Each instantaneous gauge group includes the stabilizers at
its cut together with the selected destabilizers of temporarily
and permanently masked stabilizers. Including these groups
in $\Gspace$ incorporates the gauge freedoms associated with
syndrome information that is unavailable within the measurement
window. Commuting propagators supply the relations between different
cuts. By their definition, an error before a circuit block
and its propagated counterpart after that block have the same
circuit action without changing any intervening measurement
outcome. Their product therefore represents a trivial error.
Including $\Rcomm$ in $\Gspace$ identifies these error
configurations modulo gauge.

This definition specifies
a generating set for the spacetime gauge group; the following subsections analyze the
resulting gauge group and its relation to the detectors and logical information of the circuit. 
\paragraph{Overview of the gauge group.}
Before analyzing the gauge group in detail, we summarize its structure in
Fig.~\ref{fig:ss-ingredients}; the remaining subsections of this section establish it.
The generators of $\Gspace$ are of two kinds. The detector stabilizers $\Sdet$ commute
with every element of $\Gspace$ and form its center (see Subsection~\ref{sec:gauge-decomposition}). They are the stabilizers of the
code and carry all of its syndrome information. Every other generator belongs to a
gauge qubit, a pair of operators that anticommute with each other and commute with all
other pairs. These pairs have three origins.
\begin{enumerate}
    \item \emph{Masked stabilizers.} (Definition \ref{def : IGG}) A temporarily or permanently masked stabilizer at
    cut $t$ pairs with its destabilizer at the same cut. Its syndrome cannot be learned
    from the measurement window, so its value is a gauge degree of freedom. The
    destabilizers of permanently masked stabilizers are fixed by the measurements that
    mask them. Those of temporarily masked stabilizers can be chosen, for example to
    maximize the unmasked distance.
    \item \emph{Unmasked stabilizers and commuting propagators.} (Subsection~\ref{sec:unmasked-destabilizers}) An unmasked stabilizer
    at a single cut is not a stabilizer of the spacetime code; only its combinations
    into detectors are. The remaining combinations, the unmasked stabilizers modulo
    $\Sdet$, pair with commuting propagators.
    \item \emph{Adjacent logical generators.}   (Subsection~\ref{sec:adjacent-logical-generators}) Products of representatives of the same
    logical operator at adjacent cuts lie in $\Gspace$. When the number of cuts is odd,
    they pair up into gauge qubits.
\end{enumerate}
Counting these pairs shows that what remains is exactly the $k$ logical qubits carried
by the circuit (see Subsection~\ref{sec:gauge-decomposition}).

Of these three kinds of gauge qubit, the first is already in hand: each masked
stabilizer comes paired with its destabilizer by construction of
$\atcut{\Gspace}{t}$ (Definition~\ref{def : IGG}). The other
two are not yet pairs. Section~\ref{sec:unmasked-destabilizers} shows that
commuting propagators supply partners for the unmasked stabilizers, and
Section~\ref{sec:adjacent-logical-generators} that adjacent logical generators lie in
$\Gspace$ and pair among themselves.

\subsection{Commuting propagators and unmasked destabilizers}
\label{sec:unmasked-destabilizers}

Let
\begin{equation}~\label{eq: unmask}
    \Uspace
    \coloneqq
    \bigoplus_{t=0}^{T}\atcut{\Uspace}{t}
\end{equation}
be the space generated by all unmasked
stabilizers supported on a single cut. Since each detector stabilizer restricts
to an unmasked stabilizer at every cut,
$\Sdet\subseteq\Uspace$. The space $\Uspace$ is
isotropic with respect to $\bigoplus_{t=0}^{\FinalCut}\atcut{\Gspace}{t}$.

An unmasked stabilizer at a single cut is not a stabilizer of the spacetime code; only
detector stabilizers are. Commuting propagators supply the missing destabilizers: they
can anticommute with individual cut-supported unmasked stabilizers while commuting
with every detector stabilizer. In this subsection we make this precise and give an
algorithm that constructs the pairs that form gauge qubits. \\

\begin{theorem}[Unmasked operators commuting with all propagators]
\label{thm:unmasked-propagator-kernel}
Let
$
    \Uspace
    \coloneqq
    \bigoplus_{t=0}^{\FinalCut}\atcut{\Uspace}{t}
$
be as in Eq~\ref{eq: unmask} .
Assume that the perfect terminal readout at $T+1/2$ measures a generating set of $\Stab^{(T)}$. 
Let $\Sdet$ be the detector stabilizer space defined in
Section~\ref{sec : prelim_fauls and detectors}.
Then
\begin{equation}
\Uspace\cap\Rcomm^\perp=\Sdet.
\label{eq:unmasked-propagator-kernel}
\end{equation}
In other words, the combinations of unmasked stabilizers that commute with every commuting propagator are precisely the detector stabilizers.
\end{theorem}

\begin{proof}
We first show that every detector stabilizer commutes
with $\Rcomm$. By Eq.~\eqref{eq:detector-cut-stabilizer}
and
\begin{equation}
    C_{r,t}=C_{r,t+1}\atcut{U}{t+1/2},
    \qquad r\geq t+1.
\end{equation}
For each internal block, $0\leq t<T$, the restrictions of a
detector stabilizer satisfy
\begin{equation}
    \Uelt{\delta}{t}
    =
    d_\delta^{(t+1/2)}
    \left(\atcut{U}{t+1/2}\right)^\dagger
    \Uelt{\delta}{t+1}
    \atcut{U}{t+1/2}.
    \label{eq:detector-adjacent-restrictions}
\end{equation}

Consider a commuting propagator
\begin{equation}
    P^{(t)}
    +
    \left[
        \atcut{U}{t+1/2}
        P
        \left(\atcut{U}{t+1/2}\right)^\dagger
    \right]^{(t+1)}.
\end{equation}
By definition, $P$ commutes with every measurement
in $\atcut{\ms}{t+1/2}$, and hence with their
product $d_\delta^{(t+1/2)}$.
Equation~\eqref{eq:detector-adjacent-restrictions}
and Clifford invariance of the symplectic form give
\begin{equation}
\begin{split}
    \symp\left(P,\Uelt{\delta}{t} \right) &= \symp \left( P, d_\delta^{(t+1/2)}
    \left(\atcut{U}{t+1/2}\right)^\dagger
    \Uelt{\delta}{t+1}
    \atcut{U}{t+1/2}\right) \\
    &=   \symp\left(P, d_\delta^{(t+1/2)}\right) + \symp \left( P,
    \left(\atcut{U}{t+1/2}\right)^\dagger
    \Uelt{\delta}{t+1}
    \atcut{U}{t+1/2}\right) \\
    &= \symp\!\left(
        \atcut{U}{t+1/2}
        P
        \left(\atcut{U}{t+1/2}\right)^\dagger,
        \Uelt{\delta}{t+1}
    \right)
\end{split}
\end{equation}

The two contributions to the pairing of every commuting propagators with
$\DetStab{\delta}$ over adjacent timesteps therefore cancel over $\Ftwo$. Since $\DetStab{\delta}\in\Uspace$, this proves
\begin{equation}
    \Sdet\subseteq\Uspace\cap\Rcomm^\perp.
\end{equation}

Conversely, let $v\in\Uspace\cap\Rcomm^\perp$,
with restriction $v^{(t)}\in\atcut{\Uspace}{t}$
at each cut. Commutation with every commuting
propagator across the block $t+1/2$ implies
\begin{equation}
\begin{split}
    \symp (v, P) = 0 \quad\Rightarrow\quad &\symp(\atcut{v}{t}, P) = \symp\left( \atcut{v}{t+1}, \atcut{U}{t+1/2}
        P
        \left(\atcut{U}{t+1/2}\right)^\dagger\right) \\
    &\symp(\atcut{v}{t}, P) =  \symp\left( \left(\atcut{U}{t+1/2}\right)^\dagger \atcut{v}{t+1} \atcut{U}{t+1/2}, 
        P
        \right) \\
     &\symp\left( \atcut{v}{t}\left(\atcut{U}{t+1/2}\right)^\dagger \atcut{v}{t+1} \atcut{U}{t+1/2}, 
        P
        \right) = 0.  
\end{split}
\end{equation}
for every Pauli $P$ commuting with all measurements
in $\atcut{\ms}{t+1/2}$. By nondegeneracy of the
Pauli symplectic form,
\begin{equation}
    v^{(t)}
    \left(\atcut{U}{t+1/2}\right)^\dagger
    v^{(t+1)}
    \atcut{U}{t+1/2}
    \in
    \left\langle\atcut{\ms}{t+1/2}\right\rangle.
    \label{eq:unmasked-measurement-product}
\end{equation}
For each internal block, choose a product of measurements
$d^{(t+1/2)}$ realizing Eq.~(48). At the terminal boundary,
$v^{(T)}\in U^{(T)}=\Stab^{(T)}$, so the perfect
readout provides a product $d^{(T+1/2)}=v^{(T)}$.
Iterating the resulting relations expresses the known input
value of $v^{(0)}$ as a parity of these measurement outcomes.
This parity defines a detector $\delta$ whose cut restrictions
are $u_\delta^{(t)}=v^{(t)}$. Hence
$v=s_\delta\in S_{\mathrm{det}}$, proving the reverse inclusion.

\end{proof}
The next theorem shows that we can pair up unmasked stabilizers and commuting propagators so that they form gauge qubits. \\

\begin{theorem}[Destabilizers of the unmasked stabilizers]
\label{thm:unmasked-gauge-partners}
For any representatives $u_1,\ldots,u_q$ of a basis
of $\Uspace/\Sdet$, there exist
$r_1,\ldots,r_q\in\Rcomm$ such that
\begin{equation}
    \symp(u_i,r_j)=\delta_{ij}.
    \label{eq:unmasked-dual-pairing}
\end{equation}
Every $r_j$ commutes with $\Sdet$.
\end{theorem}

\begin{proof}
We show the proof using a constructive approach, let $\{p_1, p_2, \ldots, p_m\}$ be a basis for $\Rcomm$. Let us define,
\begin{equation}
    A_{ij} = \symp(u_i, p_j).
\end{equation}

Theorem \ref{thm:unmasked-propagator-kernel} guarantees that $A$ has full rank. If rows in $A$ are dependent then there would exist some $u = \sum_i a_i u_i$ such that it commutes with every $p_j$, hence with all of $\Rcomm$. Theorem \ref{thm:unmasked-propagator-kernel} would then imply $u\in \Sdet$, contradicting independence of the quotient basis. Thus, there exists a binary matrix $C$ which satisfies the following equation,
\begin{equation}
\label{eq:constrcu_destab}
    AC = I.
\end{equation}
Then, we construct the destabilizers of the unmasked stabilizers as $r_j = \sum_k C_{kj}p_k$ which ensures that, 
\begin{equation}
\begin{split}
    \symp(u_i, r_j) &= \sum_k \symp(u_i,p_k) C_{kj} \\ 
    &= \sum_k A_{ik}C_{kj} \\
    &= [I]_{ij}\\
    &= \delta_{ij}
\end{split}
\end{equation}
\end{proof}
Note that set of $r_j$ generated by Theorem \ref{thm:unmasked-gauge-partners} need not commute with one another or with other gauge generators; Lemma~\ref{lem:canonical-unmasked-pairs} corrects this, at the cost of leaving $\Rcomm$.

\subsection{Adjacent logical generators}
\label{sec:adjacent-logical-generators}
In this subsection, we show that the product of adjacent
representatives is a gauge operator.

Choose representatives
$\{\bar X_i^{(t)},\bar Z_i^{(t)}\}_{i=1}^{k}$
of the logical Paulis at each cut.
Representatives with the same label are related
by the circuit's logical update rule. Define
\begin{equation}
    \Ladj
    \coloneqq
    \left\langle
        \bar X_i^{(t)}+\bar X_i^{(t+1)},
        \bar Z_i^{(t)}+\bar Z_i^{(t+1)}
        :
        1\leq i\leq k,\;
        0\leq t<T
    \right\rangle.
\end{equation}

\begin{lemma}
\label{lem:adjacent-logicals-in-gauge}
The adjacent logical generators belong to $\Gspace$ i.e. $\Ladj \subseteq \Gspace$.
\end{lemma}

\begin{proof}
Consider a logical representative
$\ell^{(t)}$ and the next circuit block. We can multiply $\ell^{(t)}$ by an instantaneous stabilizer $s$ so that the resulting representative $P=\ell^{(t)}+s$
commutes with every measurement in the block. Such a stabilizer is guaranteed to exist as otherwise there exists a product of measurements that commutes with every stabilizer but anti-commutes with the logical operator. This would be a logical measurement and isn't allowed according to the setup (Section \ref{sec:setup}).

The Clifford image of $P$ represents the same logical Pauli at the next cut, so
\begin{equation}
    \ell^{(t+1)}
    =
    U_tPU_t^\dagger+s'
\end{equation}
for some $s'\in\atcut{\Stab}{t+1}$.
Therefore,
\begin{equation}
    \ell^{(t)}+\ell^{(t+1)}
    =
    \left[
        P^{(t)}+(U_tPU_t^\dagger)^{(t+1)}
    \right]
    +s^{(t)}+s'^{(t+1)}.
\end{equation}
The first two terms form an element of $\Rcomm$, while
the remaining terms belong to the instantaneous
gauge groups. Thus the adjacent logical
generator lies in $\Gspace$.
\end{proof}

\begin{lemma}[Pairing adjacent logical generators]
\label{lem:adjacent-logical-pairing}
If $T+1$ is odd, then $\Ladj$ is nondegenerate
and has dimension $2kT$. It therefore supplies
$kT$ gauge qubits.
\end{lemma}

\begin{proof}
The $2kT$ adjacent logical generators are linearly independent, hence $\dim \Ladj = 2kT$. It remains to show that $\Ladj$ is non-degenerate.

Suppose $g\in\Ladj$ commutes with every adjacent logical generator.
Commutation with the adjacent $Z$-type generators forces the
$X$-coefficients of $g$ to agree at consecutive cuts, and similarly
for its $Z$-coefficients. Thus each logical coefficient is constant
across all cuts.

However, every adjacent generator contributes each logical
coefficient twice, so its sum over all cuts is zero over
$\mathbb F_2$. Since the number of cuts is odd, a constant
coefficient can have zero sum only if it vanishes.
Hence $g=0$, proving nondegeneracy.

A nondegenerate Pauli space of dimension $2kT$
forms $kT$ gauge qubits.
\end{proof}

The $X$- and $Z$-type generators
on the same pair of cuts commute, so the canonical
gauge pairs are formed by pairing suitable
linear combinations of adjacent generators.

\subsection{The center and the gauge-qubit decomposition}
\label{sec:gauge-decomposition}
We now determine the structure of the gauge group. We show that its center is exactly
the detector stabilizer space $\Sdet$, and that the gauge qubits fall into three
categories, with two according to their masking properties: masked stabilizers paired with their
destabilizers, unmasked stabilizers paired with commuting propagators, and pairs formed
from adjacent logical representatives. Counting dimensions then confirms that the code
encodes exactly the $k$ logical qubits carried by the circuit.

The partners $r_i$ of Theorem~\ref{thm:unmasked-gauge-partners} pair correctly with the
$u_i$, but they may anticommute with one another and with the other elements in the gauge group. Lemma~\ref{lem:canonical-unmasked-pairs} removes
these anti-commutations by a symplectic Gram--Schmidt step. This does not change the pairings
with the $u_i$, and the corrected $r_i$ remain in $\Gspace$, although not necessarily in
$\Rcomm$. \\

\begin{algorithm}[ht]
\caption{Completing the canonical unmasked gauge pairs}\label{algorithm:canonical basis}

\begin{algorithmic}[1]%~
\Require Operators $u_i,r_i$ from
Theorem ~\ref{thm:unmasked-gauge-partners};
a canonical paired basis
$\{h_\alpha,h_\alpha^\sharp\}_\alpha$
of $\mathsf H$
\Ensure Canonical pairs
$\{u_i,\widehat r_i\}_{i=1}^{q}$
commuting with $\mathsf H$ and $\Sdet$
\For{$i=1,\ldots,q$}
    \State Set
    \begin{equation}
        r_i'
        \gets
        r_i+
        \sum_\alpha
        \left[
            \symp(r_i,h_\alpha^\sharp)h_\alpha
            +
            \symp(r_i,h_\alpha)h_\alpha^\sharp
        \right].
    \end{equation}
    \State Set
    \begin{equation}
        \widehat r_i
        \gets
        r_i'
        +
        \sum_{j<i}
            \symp(r_i',\widehat r_j)u_j.
    \end{equation}
\EndFor
\State \Return $\{(u_i,\widehat r_i)\}_{i=1}^{q}$.
\end{algorithmic}
\end{algorithm}
\begin{lemma}[Completing the unmasked gauge pairs]
\label{lem:canonical-unmasked-pairs}
Let $H_m$ denote the set of gauge qubits for the masked stabilizers and their destabilizers. Let $\mathsf H$ be the set of canonical paired basis $\{h_\alpha,h_\alpha^\sharp\}_\alpha$ for $H_m\cup \Ladj$ such that:\begin{equation}
    \symp(h_\alpha,h_\beta^\sharp)=\delta_{\alpha\beta},
    \qquad
    \symp(h_\alpha,h_\beta)
    =
    \symp(h_\alpha^\sharp,h_\beta^\sharp)=0.
\end{equation}   Let $u_i,r_i$ be the operators constructed in
Theorem~\ref{thm:unmasked-gauge-partners}.
We apply algorithm~\ref{algorithm:canonical basis}.
Then, the resulting operators satisfy
\begin{align}
    \symp(u_i,\widehat r_j)&=\delta_{ij},\\
    \symp(\widehat r_i,\widehat r_j)&=0,\\
    \symp(\widehat r_i,h)&=0
    \qquad\text{for every }h\in\mathsf H.
\end{align}
They also commute with $\Sdet$ and belong to
$\Gspace$. Thus the pairs
$\{u_i,\widehat r_i\}$ are canonical gauge pairs
orthogonal to $\mathsf H$.
\end{lemma}

\begin{proof}
For each canonical pair
$h_\alpha,h_\alpha^\sharp$, adding
$\symp(r_i,h_\alpha^\sharp)h_\alpha$
cancels the pairing with $h_\alpha^\sharp$,
and adding
$\symp(r_i,h_\alpha)h_\alpha^\sharp$
cancels the pairing with $h_\alpha$.
This makes $r_i'$ commute with all of $\mathsf H$. These additions preserve its pairings with
$\Uspace$, because $\mathsf H$ commutes with
$\Uspace$.

Next, assume that
$\widehat r_1,\ldots,\widehat r_{i-1}$
have already been constructed.
For $\ell<i$, the next step ensures that
\begin{align}
    \symp(\widehat r_i,\widehat r_\ell)
    &=
    \symp(r_i',\widehat r_\ell)
    +
    \sum_{j<i}
        \symp(r_i',\widehat r_j)
        \symp(u_j,\widehat r_\ell)
    \nonumber\\
    &=0.
\end{align}
Adding the $u_j$ preserves all pairings with
$\Uspace$, since $\Uspace$ is isotropic, and
preserves commutation with $\mathsf H$.

Finally, the original $r_i$ commute with
$\Sdet$, and all corrections lie in
$\mathsf H+\Uspace$, which also commutes with
$\Sdet$. Every operator added belongs to
$\Gspace$, so $\widehat r_i\in\Gspace$.
\end{proof}

Lemma~\ref{lem:canonical-unmasked-pairs} completes the unmasked
operators into canonical gauge pairs, orthogonal to the masked
and adjacent-logical gauge pairs.
Theorem~\ref{thm:gauge-structure} shows that these pairs, together
with the detector stabilizers, span $\Gspace$.
Consequently, the detector stabilizers form its center, and
$\Gspace$ consists precisely of the undetectable faults with
trivial logical action. \\

\begin{definition}[Logical action of an undetectable fault]
An undetectable fault has \emph{trivial logical action} if it
preserves the values of all propagated logical Pauli observables, including the measurement-dependent sign updates. Equivalently, it induces the identity Pauli on the encoded
output state. Otherwise, it has \emph{nontrivial logical action}.  \\
\end{definition}

\begin{theorem}[Structure of the gauge group]
\label{thm:gauge-structure}
Let $\Gspace$ be the spacetime subsystem code constructed as in Def \ref{def:spacetime-gauge-group} then
\begin{enumerate}
    \item $Z(\Gspace)=S_{\mathrm{det}}$.
    \item A spacetime fault belongs to $\Gspace$ if and only if
    it is undetectable and has trivial logical action.
\end{enumerate}
\end{theorem}

\begin{proof}

Every generator of $\Gspace$ is undetectable and has trivial
logical action. Indeed, instantaneous stabilizers act trivially,
and the two faults in a commuting propagator cancel without changing the intervening measurement outcomes.
For the chosen masked destabilizers, we use the fact that they can be constructed to have trivial detector syndrome and induce no net change in the propagated logical observables, including their measurement-dependent sign updates. Note that this follows from the construction of the instantaneous gauge group using the distance algorithm in~\cite{fu2024errorcorrectiondynamicalcodes}. Thus every element of $\Gspace$ is undetectable and has
trivial logical action.

We next show that every undetectable fault is gauge-equivalent
to a logical operator supported on the final cut.
We do this by successively removing the destabilizers of unmasked stabilizers and the masked components, leaving only stabilizers
and logical representatives.

Let $\{u_i,\widehat r_i\}_{i=1}^{q}$ be the canonical pairs
from Lemma~\ref{lem:canonical-unmasked-pairs}. Define
\[
    E_1
    :=
    E+\sum_{i=1}^{q}\omega(E,u_i)\widehat r_i.
\]
Since $\omega(\widehat r_i,u_j)=\delta_{ij}$, each added
partner cancels the corresponding pairing of $E$ with $u_i$. Moreover, $E$ commutes with $S_{\mathrm{det}}$ because it is
undetectable, and the added partners preserve this commutation.
Since $S_{\mathrm{det}}$ and the $u_i$ span $U$, the operator
$E_1$ commutes with all of $U$.

Next, remove the masked components. Let
$\{h_\alpha,h_\alpha^\sharp\}_\alpha$ be a canonical paired basis
of the masked gauge space $H_{\mathrm m}$, and set
\[
    E_2
    :=
    E_1+\sum_\alpha
    \left[
        \omega(E_1,h_\alpha^\sharp)h_\alpha
        +\omega(E_1,h_\alpha)h_\alpha^\sharp
    \right].
\]
These additions cancel the pairings with both members of
every masked pair, so $E_2$ commutes with $H_{\mathrm m}$.
They also preserve commutation with $U$, since
$H_{\mathrm m}$ commutes with $U$.
All operators added so far belong to $\Gspace$.

Commutation with $U$ excludes destabilizers of unmasked stabilizers,
while commutation with both members of every masked pair
excludes masked stabilizer and destabilizer components.
Thus only unmasked stabilizers and logical representatives
remain:
\[
    E_2=u+\sum_{t=0}^{T}\ell_t^{(t)},
    \qquad u\in U.
\]
Here $\ell_t^{(t)}$ denotes the logical term at cut $t$;
different cuts may carry different logical operators.
Since $u\in\Gspace$, it can also be removed modulo $\Gspace$.

Finally, adjacent representatives of the same logical
operator differ by an element of $\Ladj\subseteq\Gspace$.
For each logical term, summing these adjacent relations gives
\[
    \ell_t^{(t)}+\ell_t^{(T)}
    =
    \sum_{v=t}^{T-1}
    \left(\ell_t^{(v)}+\ell_t^{(v+1)}\right)
    \in\Ladj.
\]
We can therefore move every logical term to the final cut
modulo $\Gspace$ and combine them into
$\ell^{(T)}:=\sum_t\ell_t^{(T)}$. Consequently,
\[
    E=g+\ell^{(T)},
    \qquad g\in\Gspace.
\]
The operator $g$ is a combination of unmasked stabilizers,
the corrected partners $\widehat r_i$, masked gauge pairs,
and adjacent logical generators.

Since $g$ has trivial logical action, $E$ has the same logical action as
$\ell^{(T)}$. If $E$ has trivial logical action, the final
logical class vanishes, so $E\in\Gspace$.
Together with the first paragraph, this proves that
$\Gspace$ consists exactly of the undetectable faults
with trivial logical action.

Applying this reduction to any $E\in\Gspace$ shows that
the constructed generators span all of $\Gspace$.
By Lemma~\ref{lem:canonical-unmasked-pairs} and the
nondegeneracy of $\Ladj$, they give the orthogonal
decomposition
\begin{equation}
    \Gspace
    =
    S_{\mathrm{det}}
    \oplus H_{\mathrm m}
    \oplus \Ladj
    \oplus
    \bigoplus_{i=1}^{q}
    \operatorname{span}\{u_i,\widehat r_i\}.
\end{equation}
Every summand except $S_{\mathrm{det}}$ is nondegenerate,
and $S_{\mathrm{det}}$ commutes with every summand.
Hence
\begin{equation}
    Z(\Gspace)=S_{\mathrm{det}}.
\end{equation}
\end{proof}

\subsection{Logical faults and distance}
\label{sec:spacetime-logical-distance}

In this section we first define the circuit distance and the subsystem spacetime-code
distance and then we show that both notions of distance are in fact exactly equivalent. \\

\begin{definition}[Circuit distance]
\label{def:circuit-distance}
The circuit distance is the minimum weight of an undetectable
fault with nontrivial logical action:
\[
    d_{\mathrm{circ}}
    :=
    \min\left\{
        \operatorname{wt}(E):
        \sigma(E)=0,\;
        E \text{ has nontrivial logical action}
    \right\}.
\]

\end{definition}

\begin{definition}[Subsystem spacetime-code distance]
\label{def:ssc-distance}
The distance of the subsystem spacetime code with gauge
group $\Gspace$ is the minimum weight of a nontrivial
dressed logical operator:
\[
    d_{\mathrm{SSC}}
    :=
    \min\left\{
        \operatorname{wt}(E):
        E\in Z(\Gspace)^\perp\setminus\Gspace
    \right\}.
\]
\end{definition}

The first definition concerns the effect of a fault on the
circuit, whereas the second uses only the gauge group of
the associated subsystem code. \\

\begin{theorem}[Equality of circuit and subsystem-code distance]
\label{thm:distance-equality}
For a subsystem space time code constructed as in Definition \ref{def:spacetime-gauge-group}
\[
    d_{\mathrm{circ}}=d_{\mathrm{SSC}}.
\]
\end{theorem}

\begin{proof}
A fault is undetectable precisely when it commutes with
every detector stabilizer. By
Theorem~\ref{thm:gauge-structure},
\[
    \sigma(E)=0
    \quad\Longleftrightarrow\quad
    E\in S_{\mathrm{det}}^\perp
    =
    Z(\Gspace)^\perp.
\]
The same theorem states that an undetectable fault has
trivial logical action precisely when it belongs to
$\Gspace$. Therefore,
\[
    \left\{
        E:\sigma(E)=0,\;
        E \text{ has nontrivial logical action}
    \right\}
    =
    Z(\Gspace)^\perp\setminus\Gspace.
\]
Minimizing the same weight over these identical sets gives
the claimed equality.
\end{proof}

The correspondence also describes the logical operators
of the subsystem spacetime code.
The proof of Theorem~\ref{thm:gauge-structure} shows that
every undetectable fault can be written as
\[
    E=g+\ell^{(T)},
    \qquad g\in\Gspace,
\]
where $\ell^{(T)}$ is a final-cut logical representative.
Thus every dressed logical class has a representative
supported on the final cut, and its circuit logical action
is exactly that of this representative.
This reduction need not preserve weight: a representative
distributed across several cuts may have smaller weight.
The distance therefore requires minimization over the
whole spacetime fault space. \\

\begin{definition}[Gauge equivalence]
\label{def:gauge-equivalence}
Two spacetime faults $E$ and $E'$ are gauge-equivalent if
\[
    E+E'\in\Gspace.
\]
\end{definition}

Gauge-equivalent faults have the same detector syndrome.
Conversely, if two faults have the same syndrome, their
difference is undetectable, and they are gauge-equivalent
precisely when this difference has trivial logical action.
Given an actual fault $E$ and an estimated fault
$\widehat E$ with the same syndrome, the residual fault
$E+\widehat E$ is undetectable. The estimate succeeds
precisely when
\[
    E+\widehat E\in\Gspace.
\]
A decoder therefore need not identify the exact fault;
it must identify its gauge-equivalence class.

The subsystem spacetime code therefore captures both
which faults are detected and which undetectable faults
act nontrivially on the encoded information.
Its distance measures the minimum weight of such a
nontrivial undetectable fault over the entire measurement
window.

\subsection{Beyond the steady state set-up}
\label{sec:beyond_steady_state}
Section~\ref{sec: SS code} assumed the
set-up of Section~\ref{sec:setup}: a fixed number $k$ of logical qubits carried
across the whole window, no logical measurements, and a perfect terminal readout of
a generating set of $\Stab^{(T)}$. In this subsection we drop these assumptions one
at a time. The construction of the gauge group (Definition~\ref{def:spacetime-gauge-group})
remains largely unchanged: The instantaneous gauge groups and the
commuting propagators are still in $\Gspace$.

\subsubsection{Partial terminal round}~\label{subsub: partial terminal round}
Beyond the steady state, the perfect terminal readout may contain measurements that are not generators of $\Stab^{(T)}$. Thus, the measurements may permanently masked some of the generators of $\Stab^{(T)}$, with destabilizers given by the terminal measurements.
Because of the partial terminal readout, some of stabilizers
of the final ISGs that could have been unmasked at the terminal round are treated as \emph{temporarily masked}. The last measurements can also
create \emph{permanently masked} stabilizers, whose syndromes are lost
irreversibly. Both kinds can be treated using the formalism developed for an instantaneous gauge group.

\subsubsection{Logical measurements}
\label{sec:setup-logical}

Suppose a logical Pauli
$\bar L$ is measured with outcome $O(\bar L)$. Then $(-1)^{O(\bar L)}\bar L$
becomes a stabilizer of the post-measurement state. The rank of the ISG
increases by one and the number of logical qubits decreases by one at all
later cuts. The outcome itself is a logical outcome and not a detector:
its value is not predicted by earlier outcomes. A later measurement of
$\bar L$ (or of an operator equal to it modulo the ISG) is predicted, and
does give a detector. 

At cuts before the logical measurement, the logical qubit is simply treated as a logical qubit in the subsystem spacetime code formalism. The product of adjacent logical representatives is added to the gauge group, till it is measured. After the logical measurement at $t+1/2$, the number of logical qubits decreased for all later cuts, and $\bar{L}$ is treated as a stabilizer in each ISG.

% ---------------------------------------------------------------------------
\subsubsection{Reset and discard}
\label{sec:setup-reset}

Discarding a qubit ends its ``wire''. At every later cut there is one fewer location, and Pauli components on that qubit are dropped. The ISG is updated by keeping the stabilizers that act trivially on the discarded qubit, after multiplying generators to clear their support there where possible. Those that cannot be cleared become permanently masked. Detectors are back-propagated only onto the qubits present at each cut. A reset is a discard followed by the preparation of a fresh qubit at the next cut. The fresh qubit enters with its initial stabilizer and at a new spatial location, $Z$ for $\ket 0$ and $X$ for $\ket +$, which becomes an element of the ISG.

\subsubsection{State preparation}~\label{subsec: state prep}
A prepared logical state enters as $\bar Z\in\atcut{\mathsf S}{0}$ with a known
sign. If $\bar Z$ is measured later in the window, it is an ordinary unmasked
stabilizer and contributes a detector. If it is never measured, we treat it as
a stabilizer at every cut, $\atcut{\bar Z}{t}\in\atcut{\mathsf S}{t}$ and exclude $\bar Z$ from the
terminal readout. Equivalently, $\Gspace_{\mathrm{prep}}=\langle \Gspace,
\atcut{\bar Z}{T}\rangle$ with $\Gspace$ the gauge group of
Definition~\ref{def:spacetime-gauge-group}.

The commuting propagators of $\bar X$ remain in $\Rcomm$ and anticommute with
each $\atcut{\bar Z}{t}$, so the individual $\atcut{\bar Z}{t}$ are gauge
operators, not stabilizers. The center is
\begin{equation}
  \Z{\Gspace_{\mathrm{prep}}}=\Sdet\oplus\langle \bar Z_{\mathrm{ST}}\rangle,
  \qquad
  \bar Z_{\mathrm{ST}}:=\sum_{t=0}^{T}\atcut{\bar Z}{t},
\end{equation}
and the remaining $\atcut{\bar Z}{t}$ pair with the $\bar X$ propagators into
$T$ gauge qubits. The operator $\bar Z_{\mathrm{ST}}$ is the detector
stabilizer that a terminal readout of $\bar Z$ would produce. Removing the
readout leaves it in the center but unmeasured (so it is not the same as a detector stabilizer). A logical $\bar X$ fault anticommutes
with it and is undetectable, so the relevant distance is
\begin{equation}
  d_{\mathrm{prep}}
  =\min\{\mathrm{wt}(E):E\in\orth{\Sdet}\setminus\Gspace_{\mathrm{prep}}\},
\end{equation}
which differs from $\orth{\Z{\Gspace_{\mathrm{prep}}}}\setminus
\Gspace_{\mathrm{prep}}$ precisely by the $\bar X$-type faults.

\begin{remark}
$\bar Z$ is temporarily masked in the sense of
Section~\ref{subsub: partial terminal round}, but we deliberately do not gauge it
out: its destabilizer is a logical error.
\end{remark}
\subsubsection{Frameworks for measurements}~\label{subsec: Frameworks for measurements}
In this section, we give two other ways to represent measurement errors, in addition to Section \ref{sec:setup-frameworks}. They each assign different weights to measurement
errors, and so they can give different spacetime distances. Table~\ref{tab:frameworks}
summarizes them.
The three frameworks below differ only in how a measurement is represented
in the spacetime code. 

\begin{table}[ht]
\centering
\begin{tabularx}{\linewidth}{@{}p{3.8cm}p{3cm}p{3.8cm}p{3.8cm}@{}}
\toprule
 & (a) Pauli pairs & (b) measurement locations & (c) syndrome extraction circuits \\
\midrule
Fault locations& Qubits & Qubits with measurement locations at $t+1/2$ & Qubits and ancillas \\
Measurement error & $P^{(t)}P^{(t+1)}$, weight 2 & $X_m^{(t+1/2)}$, weight 1 & Error on ancilla, weight 1 \\
Support of a $s_\delta$ & On data qubits & $s_\delta$ extended by $Z_m$ & On both ancilla and data qubits \\
\bottomrule
\end{tabularx}
\caption{The three ways of representing measurements.}
\label{tab:frameworks}
\end{table}

The next two frameworks are different ways to treat measurements, either as a location on its own or as a full syndrome extraction circuit, which can itself incorporate measurement locations. We explain how these can be incorporated into the spacetime code framework although we have only used framework (a) for the earlier sections.

\begin{figure}[ht]
    \centering
    \includegraphics[width=0.4\linewidth]{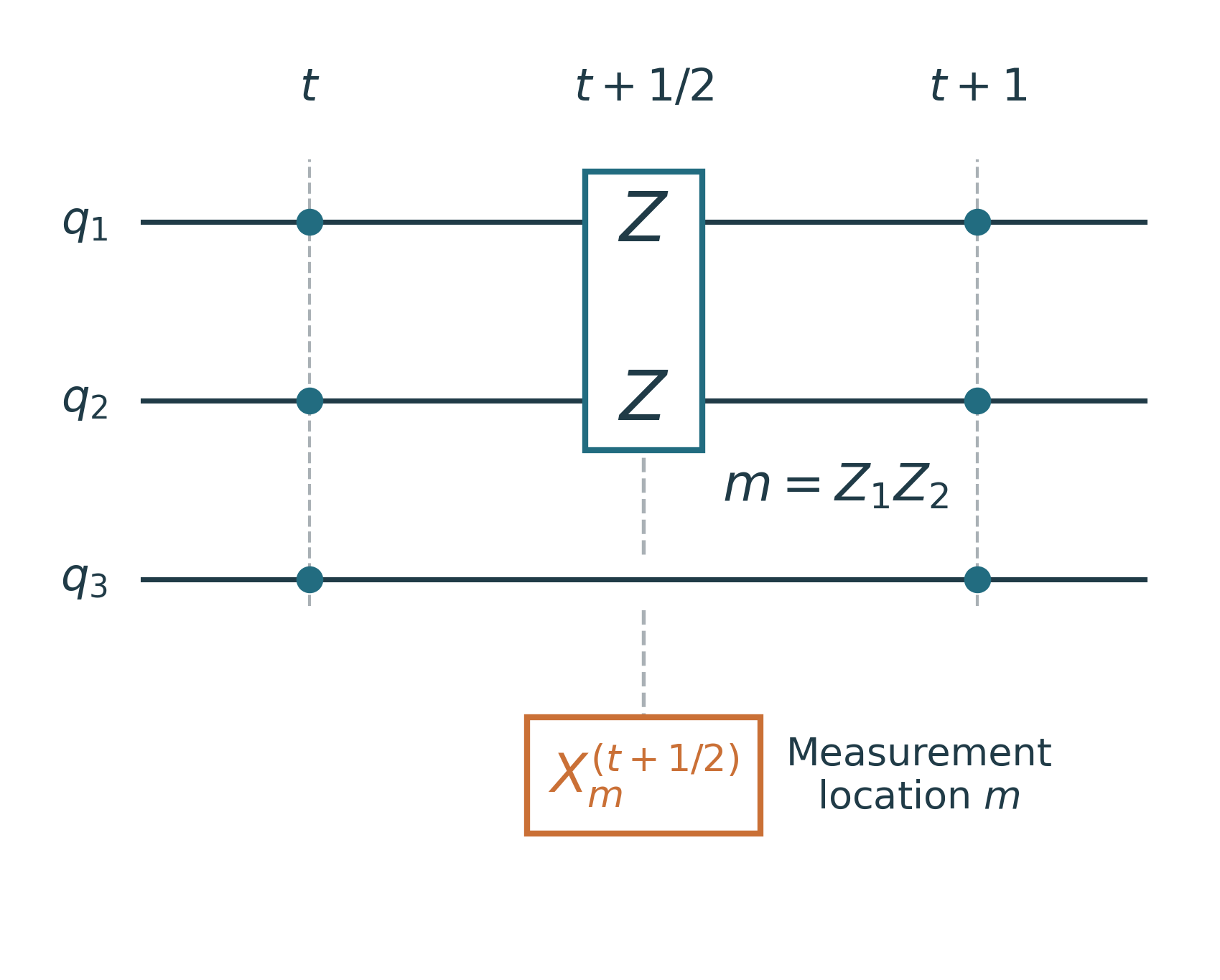}
    \caption{A measurement location $m$ for a $Z_1Z_2$ measurement. The
    operators $(\atcut{X_1}{t} + \atcut{X_1}{t+1} + \atcut{X_m}{t+1/2})$ and $\atcut{Z_m}{t+1/2}$ are added to the gauge group, and any
    gauge operator $g$ that anticommutes with $(\atcut{X_1}{t} + \atcut{X_1}{t+1} + \atcut{X_m}{t+1/2})$ is replaced by
    $g\oplus \atcut{Z_m}{t+1/2}$. The error $\atcut{X_m}{t+1/2}$ on the measurement location is then gauge equivalent to $(\atcut{X_1}{t} + \atcut{X_1}{t+1})$.}
    \label{fig:formalism_b_demo}
\end{figure}

\paragraph{Framework (b): measurement locations}Here every measurement receives a spacetime location of its own. The benefit of this is that each
measurement error is a single-location fault, whatever the weight of the
measured Pauli. Note that measurement locations also connect the code to statistical-mechanics models: they make each measurement error an independent fault, as in the random-plaquette gauge model of Ref.~\cite{dennis2002topological}, which Pauli pairs cannot always do~\cite{ aitchison2025spacetimespinsstatisticalmechanics}. For a measurement $m$ in the block $t+1/2$ we add a location
labelled $m$ at $t+1/2$. Choose a propagator $P^{(t)} + P^{(t+1)}$ 
with $\operatorname{supp}_{\mathrm{sp}}(P)\subseteq\operatorname{supp}_{\mathrm{sp}}(m)$ (see Definition~\ref{eq:spatial-support}) and
$\symp(P,m)=1$. Then, we do the following:
\begin{enumerate}
  \item Add $(P^{(t)}+ P^{(t+1)}+ X_m^{(t+1/2)})$ and $Z_m^{(t+1/2)}$ to
        $\mathcal{G}$, which form one new gauge qubit.
  \item Replace every existing generator $g$ that anticommutes with
        $P^{(t)} + P^{(t+1)}$ by $g\oplus Z_m^{(t+1/2)}$, so that all
        generators commute with the new gauge qubit.
\end{enumerate}
After step~1, $X_m^{(t+1/2)}$ is gauge equivalent to $P^{(t)} + P^{(t+1)}$, so a
readout error has the syndrome of the Pauli pair it replaces. Explicitly,
suppose the flipped outcome belongs to a detector $\delta$. Then $s_\delta$ anticommutes
with $P^{(t)} + P^{(t+1)}$, so step~2 turns it into $s_\delta\,Z_m^{(t+1/2)}$ which anti-commutes with $X_m^{(t+1/2)}$, as required. The choice
of $P$ does not matter, since two valid choices differ by an element of
$\mathcal{G}$ (see also Figure~\ref{fig:formalism_b_demo}).

\paragraph{Framework (c): syndrome extraction circuits}
To take into account potential error propagation due to say hook errors, here each measurement is replaced by its syndrome extraction circuit. In circuit designs, this is important for verifying the distance of novel measurement schemes. 
The number of qubits changes from cut to cut due to the presence of fresh ancilla qubits. A measurement error is now
an ordinary fault on the ancilla just before its readout, of weight 1.

\paragraph{Construction.}
The construction of
Section~\ref{sec:setup-basic} applies to the circuit, with the following
changes:
\begin{itemize}
  \item The preparation of an ancilla adds its initial stabilizer ($Z_a$ for $\ket 0$, $X_a$ for $\ket +$) to the ISG.
  \item Once the ancilla qubit is measured, it is then discarded and does not appear on future locations in spacetime. 
\end{itemize}

The detectors now have support on locations within the measurement circuit and
on the ancillas. Restrictions to a single cut still give the unmasked stabilizers.

\paragraph{Example.}

\begin{figure}[ht]
\centering
\includegraphics[width=.4\linewidth]{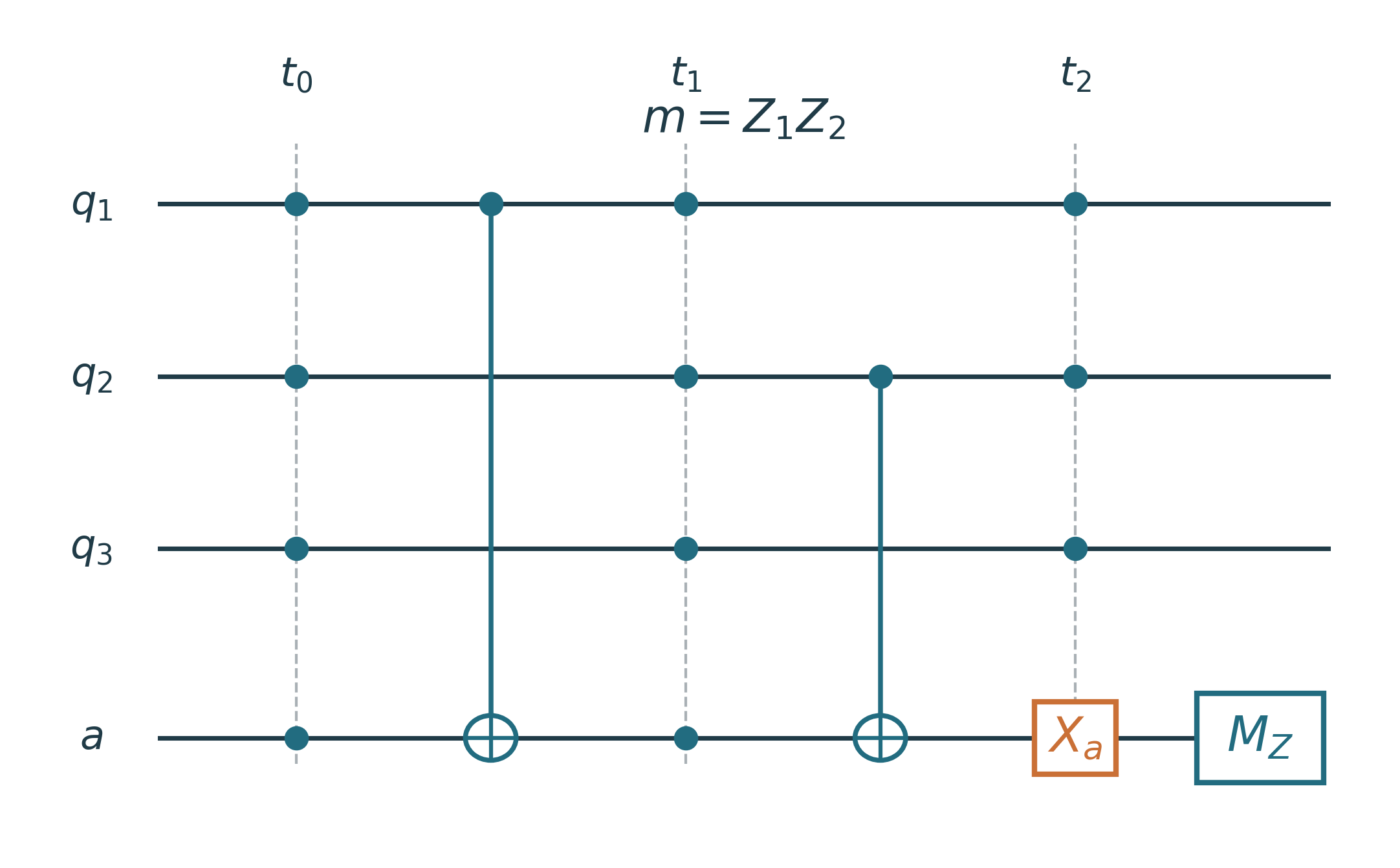}
  \caption{This sub-circuit shows the syndrome extraction circuit for  $Z_1 Z_2$ using an ancilla qubit.}
  \label{fig:meas_circuit}
\end{figure}

Consider the measurement of $Z_1Z_2$ in Figure~\ref{fig:meas_circuit}. The
ancilla $a$ is prepared in $\ket 0$, then $\mathrm{CNOT}_{1\to a}$ and
$\mathrm{CNOT}_{2\to a}$ are applied, and $Z_a$ is read out. The cuts are
$t=0,1,2$ and at each cut $\atcut{\Stab}{t} = \langle \atcut{s}{t}, \atcut{s_{\text{ancilla}}}{t}\rangle$. The ISG evolves as
\begin{equation}
  \atcut{\Stab}{0} = \langle Z_1Z_2,\ Z_a\rangle \;\to\; \atcut{\Stab}{1} = \langle Z_1Z_2,\ Z_1Z_a\rangle
  \;\to\; \atcut{\Stab}{2} = \langle Z_1Z_2,\ Z_1Z_2Z_a\rangle ,
\end{equation}
and the readout gives $Z_a = s^{(3)}\cdot s_{\mathrm{ancilla}}^{(3)}$. Back-propagating $Z_a$ gives
the stabilizer for the detector
\begin{equation}
  s_\delta \;=\; (Z_1Z_2Z_a)^{(0)} + (Z_2Z_a)^{(1)} + (Z_a)^{(2)} .
\end{equation}
The destabilizers for the unmasked stabilizers 
are the propagators
\begin{equation}
  X_a^{(0)} + X_a^{(1)} \quad\text{and}\quad X_a^{(1)} + X_a^{(2)} .
\end{equation}
The destabilizers anticommute with the unmasked stabilizer at each timestep but commutes with $s_\delta$. The second stabilizer generator at $t = 2$, $s= Z_1Z_2$ remains a temporarily masked stabilizer.\\ 

\begin{remark}
In all three frameworks, having the final round of perfect measurements is crucial in avoiding a low-weight undetectable error that occurs right at the terminal boundary. 
\end{remark}

\section{Relations between two distances}
\label{sec:two-distances}
The distance of the subsystem spacetime code is a property of the whole circuit,
whereas the unmasked distance is a property of one instantaneous stabilizer group at a
time and can be much cheaper to compute. In this section we ask when the first can be read
off from the second. For measurement-only circuits we identify a correctability
condition (see Condition 1), satisfied by static stabilizer codes and by several dynamical codes, including Floquet-Bacon-Shor code~\cite{Anonymous_2025}, the Floquet color code, planar Floquet code and the x+y Floquet code, under which
the spacetime distance equals $\min\{d_U,d_{\mathrm{meas}}\}$. Here
$d_{\mathrm{meas}}$ counts the faults needed to corrupt the logicals from the logical update rule without
tripping a detector. Appendix~\ref{sec: counter} shows that the equality can fail without this condition.

\textbf{Set-up} For simplicity of analysis, we consider a \emph{measurement-only} circuit with timeslices $t=0,\ldots,T$, initial stabilizer group $\Stab^{(0)}$ and a measurement schedule $\mathfrak{M}=\{\ms^{(1/2)}, \ms^{(3/2)},\ldots,\ms^{(T-1/2)} \}.$ Each round $\ms^{(t+1/2)}$ is a commuting family of Pauli measurements and each
qubit participates in at most one measurement per time slice. Faults are Pauli errors on the spacetime locations. Let $d_U$ denote the minimum unmasked distance over all the ISGs. \\

\begin{lemma}[Pushing forward]
\label{lem:pushforward}
Let $E$ be a Pauli operator at cut $t$. Then $E^{(t)}$ can be seen as an element on time $t+1$ and some measurement errors at cut $t+1/2$.
\end{lemma}
\begin{proof}
This is simply a rewrite: For example $E^{(t)} =(E^{(t)}+E^{(t+1)})+E^{(t+1)}$. The first term can be rewritten as elements in the gauge group and a set of measurement errors.
\end{proof}
Given a spacetime error $E$, we use both commuting and anti-commuting propagators to push forward its Pauli operators by applying Lemma~\ref{lem:pushforward}. According to definition, commuting propagators are part of the gauge group whereas anti-commuting propagators represent measurement errors in the intermediate circuit block.
Thus, up to gauge equivalence (see Definition~\ref{def:gauge-equivalence}), the resultant errors consist of a \emph{residual error} $E_T$, which is the errors on the final time slice and a set $M(E)$ of measurement errors. 

Next, we introduce Condition 1, under which there is a clear relation between the unmasked distance and the distance of the subsystem spacetime code.
\paragraph{Condition 1}~\label{correctability condition} A subsystem spacetime code satisfies Condition 1, if for every fault $E_T$ at timeslice $T$ with $|E_T|< d_U$, either $E_T$ is gauge equivalent to trivial errors, or $E_T$ is detectable, i.e. produces a nontrivial syndrome, by one or more perfect final rounds of measurements, regardless of the measurement errors occurring at timesteps $t< T$. In other words, $E_T$ is detectable regardless of what the measurement errors are on the earlier spacetime locations. 

The above condition does not hold in general for arbitrary dynamical codes. This works for cases where a perfect round of measurements of all the stabilizers at the last round always suffices to correct all the residual errors. As a special case, Condition 1 holds if there exists a set of permanent stabilizers that remain in every ISG and are always unmasked such that if we treat the other generators as permanently masked stabilizers, then the unmasked distance of every ISG is at least $d_U$. 

We now define a distance for a circuit with only \emph{measurement errors}.
The following definition requires the logical update rules stated in Procedure~\ref{lem:logical-update}. \\
\begin{definition}
For each logical representative $L_\ell$, let $V_\ell$ denote the set of measurements required by the logical update rule to dynamically update the logical representative from $t=0$ to $t=T$.
Let $\sigma(E)$ be the syndrome of $E$ which is a set of measurement errors. Define
\begin{equation}
d_{\mathrm{meas}}
=
\min\left\{
\operatorname{wt}(E):
\begin{array}{l}
\sigma(E)=0, \, E \text{ is gauge equivalent to a set of measurement errors.}\\
|M(E)\cap V_\ell|=1\pmod 2
\text{ for some logical operator }L_\ell
\end{array}
\right\}.
\end{equation}
If there is no such error, set $d_{\mathrm{meas}}=\infty$.
\end{definition}
This is well-defined and does not depend on the choice of logical representative: An error $E$ remains an undetectable logical error regardless of the choice of logical representative. Furthermore, by Lemma~\ref{lem:Vell-well-defined}, any two sets of measurements $V_{\ell_1}$ and $V_{\ell_2}$ are the same up to measurements and final stabilizers that form detectors. Since we do not have measurement errors for the stabilizers at the final cut, and $E$ is undetectable by detectors, $|E\cap V_{\ell_1}| =1 \pmod 2$ implies $ |E\cap V_{\ell_2}| =1 \pmod 2$.

The weight as in $\mathrm{wt}(E)$ in this definition is the weight of the original spacetime error $E$, not the number of measurement errors in $M(E)$. A pair of spacetime errors can be gauge equivalent to multiple measurement errors. Thus, $d_{\mathrm{meas}}$ is not, in general, the same as the number of measurement errors causing an undetectable logical error. 

\begin{theorem}~\label{thm:two-distances}
If \textbf{Condition 1} holds, then for a measurement-only circuit, the distance of the subsystem spacetime code $d$ is given by \begin{equation}
d=\min\{d_U,d_{\mathrm{meas}}\}.
\end{equation}
\end{theorem}
\begin{proof}
\emph{Upper bound.} This direction is obvious. We know there exists a logical error of weight $d_U$ so $d\leq d_U$ Every fault counted by
$d_{\mathrm{meas}}$ is an undetectable logical fault, so $d\le d_{\mathrm{meas}}$.

\emph{Lower bound.} Let $E$ be an undetectable fault with $\operatorname{wt}(E)<d_U$. Pushing
$E$ forward gives
\begin{equation}
    E \equiv E_T + M(E) \pmod{\Gspace},
\end{equation}
and $\operatorname{wt}(E_T)\le\operatorname{wt}(E)<d_U$, because in a measurement-only circuit
pushing an error forward does not spread its support. Since $E$ is undetectable, $E_T$
together with the measurement errors $M(E)$ trips no detector, including those of the
perfect final round(s). By Condition~1, $E_T$ is gauge equivalent to a trivial error. If $E$ is a logical fault, it must be due to $M(E)$ which must be lower bounded by
$d_{\mathrm{meas}}$, so $\operatorname{wt}(E)\ge d_{\mathrm{meas}}$. Hence every logical fault
has weight at least $\min\{d_U,d_{\mathrm{meas}}\}$, and $d=\min\{d_U,d_{\mathrm{meas}}\}$.
\end{proof}

\begin{lemma}[Persistent stabilizers]
\label{lem:persistent}
Suppose there exists a set of  stabilizers $\mathsf K\subseteq\Stab^{(0)}$ that
remain in every instantaneous stabilizer group and are always unmasked; equivalently, every
$k\in\mathsf K$ commutes with every measurement of the circuit. Let $d_K$ be the distance defined with $K$ as unmasked stabilizers and the rest as permanently masked stabilizers and their destabilizers. Then
\begin{equation}
    \min\{d_K,d_{\mathrm{meas}}\}\ \le\ d\ \le\ \min\{d_U,\ d_{\mathrm{meas}}\}.
\end{equation}
In particular, if $d_K\ge d_U$, then Condition~1 holds and $d=\min\{d_U,d_{\mathrm{meas}}\}$.
\end{lemma}
\begin{proof}
The upper bound is the first half of the proof of Theorem~\ref{thm:two-distances}, which
does not use Condition~1.

Let $E$ be an undetectable fault with $\operatorname{wt}(E)<d_K$, and push it forward:
$E\equiv E_T+M(E)\pmod{\Gspace}$ with $\operatorname{wt}(E_T)\le\operatorname{wt}(E)$.
As $\operatorname{wt}(E_T)<d_K$, the definition of $d_K$ implies that $E_T$ is in the gauge group, and therefore we are left with $M(E)$. Hence $E$ is either harmless or a
logical fault counted by $d_{\mathrm{meas}}$. If $d_{\mathrm{meas}}\geq d_K$, then $E$ is harmless and $d \geq d_K=\min\{d_K,d_{\mathrm{meas}}\}$. If $d_{\mathrm{meas}}<d_K$, then a logical error can be caused by measurement errors from $E$, and $d = d_{\mathrm{meas}}=\min\{d_{\mathrm{meas}},d_K\}$. This proves the lower bound.

If $d_K\ge d_U$, we get back Condition~1.
\end{proof}

\begin{remark}
    With Clifford layers, pushing forward in the sense of Lemma~\ref{lem:pushforward} can spread the support of a propagated Pauli operator, so $\text{wt}(E_T) \leq \text{wt}(E)$ fails in the proof of Theorem~\ref{thm:two-distances}. The theorem still holds for a stronger condition that requires every residual error from a fault of weight less than $d_U$ to be detectable. \\
\end{remark}

\begin{remark}
In general, without Condition 1, it is possible to create stabilizer codes where there exists a logical spacetime error that has smaller weight than the minimum unmasked distance. A simple counter-example is in the Appendix, where we present a randomly generated 12-qubit code satisfying the following properties: It has a basis of logical representatives that does not evolve with time. The minimum unmasked distance is $4$, but because some of the unmasked stabilizers do not persist, we can find a weight 3 undetectable spacetime error, which is a logical error. 
\end{remark}

\section{Applications}~\label{sec: appli}
We now use the framework for two practical tasks. First, we apply the decomposition of
Section~\ref{sec:two-distances} to existing dynamical codes. It locates where distance
is lost: in the instantaneous codes or in the measurement schedule. Second, we ask
what the distances of two circuit pieces guarantee (and when we can get an equality) once the pieces are glued
together, and identify potentially new logical errors that can appear because of the gluing at the boundary. We prove a distance upper bound in Section~\ref{sec: circuit composition} and an equality under the extra correctability condition in Theorem~\ref{thm:composition}. In Subsection~\ref{subsub: Distance lower bound}, we show that the distance of the composed circuit $C$
can be determined by the distances of the individual smaller circuits and the distance of a finite window around the glued boundary.

\subsection{Diagnosing distance in dynamical codes}
Computing the spacetime distance by exhaustive search is expensive, and the resulting
number alone does not say how to improve a circuit. The decomposition of
Section~\ref{sec:two-distances} gives a cheaper and more informative route.

For a measurement-only circuit, Theorem~\ref{thm:two-distances} splits the
spacetime distance into two contributions if \textbf{Condition 1} holds,
\begin{equation}
    d = \min\{d_U,\, d_{\mathrm{meas}}\}.
\end{equation}
The first, $d_U$, is a property of the instantaneous stabilizer groups and
their masking: it is the smallest weight of an undetectable logical error at a
single cut. The second, $d_{\mathrm{meas}}$, measures how many faults are needed
to corrupt the logical update rule without tripping a detector. It is governed
by the sets $V_\ell$ of outcomes used to update each logical representative,
and by whether those outcomes also enter a detector. Computing both quantities
therefore shows how large the distance is and where distance can be lost.
If $d < d_U$, the instantaneous codes are good and the schedule is at fault;
if $d = d_U$, the bottleneck is the instantaneous code itself. We illustrate
this on two dynamical codes from the literature.

\paragraph{Planar Floquet code.}
We consider the 64-qubit hexagonal patch of Vuillot's planar Floquet
code~\cite{vuillot2021planarfloquetcodes, blackwell2025code}, which encodes one logical qubit. The window starts in the steady-state instantaneous stabilizer group after a
green round, followed by the rounds R, B, G, R. Every
instantaneous stabilizer group has unmasked distance $d_U = 5$, and Condition~1
holds, since the persistent unmasked stabilizers admit no logical operator of
weight at most six. Nevertheless the spacetime distance is $d = 2$.

Since the code satisfies Condition 1, we just need to calculate $d_{\mathrm{meas}}$. The minimum-weight logical
fault is
\begin{equation}
    E = Z_1^{(0)} + Z_1^{(1)},
\end{equation}
a single flipped outcome of the weight-one check $X_1$ on a corner of the patch.
This outcome enters the update
of both logical representatives, but it belongs to no detector, because the next round measures $Z_1Z_2$, which removes $X_1$ from the stabilizer group before it is measured again. A single measurement error therefore flips
both logical outcomes undetected, and $d= \min\{d_U,d_{\mathrm{meas}}\} = 2$.

\paragraph{Floquet--Bacon--Shor code.}
We next consider the schedule of Alam and Rieffel~\cite{Anonymous_2025},
which produces a dynamical logical qubit on a $d_{\mathrm{BS}}\times d_{\mathrm{BS}}$
Bacon--Shor lattice. The code encodes $k = 2$
logical qubits. Condition 1 holds, and we can verify that $d = d_U = d_{\text{meas}}$.

These two examples show the potential applications of the two types of distances for finding new dynamical codes. The unmasked
distance is a static computation, one instantaneous stabilizer group at a time;
the sets $V_\ell$ come directly from the logical update rule; and $d_{\mathrm{meas}}$ is the distance of classical codes. This simplifies the calculation for many existing dynamical codes in the literature.

\subsection{Composition of circuits}~\label{sec: circuit composition}
Large fault-tolerant protocols can be assembled from smaller pieces whose distances have
been computed separately.
In this subsection we determine an upper bound on the distance of the composed circuit, and what must be computed to determine the distance of the composite circuit.

\paragraph{Setup.}
Let $C_1$ be a circuit on cuts $0,\dots,T_1$ with input stabilizer group $\Stab^{(0)}$
and output instantaneous stabilizer group $\Stab^{(T_1)}$, and let $C_2$ be a circuit on
cuts $T_1,\dots,T_2$ whose input stabilizer group is the same $\Stab^{(T_1)}$. We choose the input
logical representatives of $C_2$ to be the output logical representatives of $C_1$, so
that the two pieces carry the same logical qubits across the boundary. $C_1$, taken in isolation, has distance $d_1$ as a subsystem spacetime code, and $C_2$ (again in isolation) has distance $d_2$.

The composite $C = C_1\circ_{T_1} C_2$ consists of the blocks of $C_1$ followed by those
of $C_2$. 

The gluing at $T_1$ means that we have to recheck if there are elements in the temporarily masked stabilizer cosets $\atcut{\Tspace}{t} + \atcut{\Uspace}{t}$ for $t\leq T_1$ which are unmasked or permanently masked for the enlarged measurement window ending at $T_2$. In the calculation of the subsystem spacetime code distance, these temporarily masked stabilizers are all treated as unmasked stabilizers with a terminal round of measurements. We denote the newly formed cosets of permanently masked stabilizers at $t$ by $\atcut{\Pspace}{t} $, and the unmasked stabilizers at $t$ by $\atcut{\Uspace}{t}$ with respect to the composite circuit $C$. 

Gluing has four effects at the boundaries.
\begin{enumerate}
    \item \emph{Internal detectors.} Detectors from $C_1,C_2$ that are not supported at $T_1$ remain in $C$. 
    \item \emph{Glued detectors.} For $u\in\atcut{\Uspace}{T_1}$, a detector of $C_1$ is formed from a terminal readout at $T_1 +1/2$. For $C_2$ there exists $s_\delta\in \Sdet(C_2)$ such that it is $u_\delta^{(T_1)} =u$ when restricted to $T_1$. These two detectors are merged at $T_1$ and now form a single detector of $C$. 
    \item \emph{Lost detectors.} $p\in \atcut{\Pspace}{T_1}$ represents a nonzero class of stabilizers that are permanently masked in $C$. This is previously involved in a detector by using a terminal readout in $C_1$, but it is now permanently masked.
    
    \item \emph{Logical update.} Logical representatives of $C_1$ at $T_1$ need not commute with measurements at $T_1 +1/2$ and are updated using logical update rules. $V_\ell$ now consists of the measurements involved in the logical updates from $C_1$ and $C_2$, labeled $V_{\ell,C_1}$, $V_{\ell,C_2}$, the new updates at $T_1$, and potentially measurements coming from elements in $\Stab^{(T_1)}$.
\end{enumerate}

Because of the composition of the two circuits, some of the stabilizers that were originally unmasked by the terminal measurement at $T_1$ are now masked. Thus, less syndrome information may be available for $C_1$ and the unmasked distance for $C_1$ can decrease. On the other hand, the unmasked distance for $C_2$ remains unchanged. When
$\atcut{\Pspace}{T_1}$ is not trivial, the unmasked distance must be recomputed at the cuts $t\le T_1$.  \\
\begin{definition}
    Let $d_1(C)$ denote the distance of the subsystem spacetime code $C_1$ with a classification of the stabilizers at $T_1$ based on whether they are unmasked or permanently masked in $C$. In other words, the terminal readout for $C_1$ no longer measures the full set of stabilizers but instead only measures the subset of stabilizers that are classified as unmasked stabilizers with respect to the composed circuit $C$. In general, $d_1(C)\leq d_1$.
\end{definition}
\paragraph{Distance upper bound.}
The distance $d(C)$ of the circuit $C$ is upper-bounded by the following
\begin{equation}
    d(C)\le\min\left\{ d_1(C),\ d_2,\ d_{\mathrm{meas}}(C) \right\}.
    \label{eq:glue-upper}
\end{equation}
In addition, $\min\{d_1(C),d_2\}\leq \min\{d_U(C_1), d_U(C_2)\}$ and $d_{\text{meas}}(C)\leq \min\{d_{\text{meas}}(C_1),d_{\text{meas}}(C_2)\}$.

\paragraph{Faults.} Every fault that is undetectable and non-trivial in $C_1$ and $C_2$ respectively is undetectable and non-trivial in $C$. In addition, there may be faults in $C$ that are individually detectable when restricted to $t\leq T_1$ and $t\geq T_1$ but became an undetectable logical error when combined. 

Even if $\atcut{\Pspace}{T_1}$ is trivial, in general we cannot guarantee that $d(C) = \min\{d_{\text{meas}}(C),d_1(C), d_2\}$. Consider the counterexample in appendix~\ref{sec: counter} with $T_1 = 10$: $d_1=d_2 = 4$ and $d_{\mathrm{meas}}\geq 4$ and $\atcut{\Pspace}{T_1}$ is trivial but $d(C) = 3$.

The following Condition 2 is a modified Condition 1 to account for the potential spread of errors from Clifford unitaries.
\paragraph{Condition 2} Let $E$ be a fault and $E_T$ be its residual error on time slice $T$. A subsystem spacetime code satisfies Condition 2 if $|E|< \min\{d_1(C), d_2\}$ implies that the residual error $E_T$ given a subsequent perfect final round(s) of measurements is detectable regardless of measurement errors that occurred earlier, unless $E_T$ is a trivial error.  \\

\begin{theorem}
\label{thm:composition}
If $C$ satisfies Condition 2, then $d(C) = \min\{d_1(C), d_2\}$ if and only if  $d_{\text{meas}}(C) \geq \min\{d_1(C), d_2\}$.
\end{theorem}

\begin{proof}
Similar to proof for Theorem~\ref{thm:two-distances}. Condition 2 sidesteps the issue that $\mathrm{wt}(E_T)\leq \mathrm{wt}(E)$ breaks, allowing the previous proof to work here.
\end{proof}

Another way to interpret why gluing may reduce distance in general is because the subsystem spacetime code of $C_2$ assumes knowing the correct syndrome of the initial set of stabilizers. Thus, undetectable errors are restricted to a smaller subset, where there are no initial measurement errors. However, for the composed circuit $C$, we need to calculate error chains where the eigenvalues of stabilizers at $T_1$ can be wrong due to measurement errors earlier from $C_1$. 

For $C_1$ the same issue arises at the boundary, even after taking into account the fact that some of the stabilizers are now permanently masked. The measurement errors at $T_1 +1/2$ are also not included in the distance calculation of its spacetime code. 

So we have the situation where we can now have measurement errors at the boundary, such that they create new undetectable error patterns. One such scenario arises where a stabilizer $s$ in $C_2$ may stem from a newly initialized stabilizer in $C_1$ as a destabilizer to some permanently masked stabilizer. If it does not persist for long enough but is still used as a detector, one can easily create a short string of undetectable measurement errors.

\subsubsection{Distance lower bound}~\label{subsub: Distance lower bound}
Suppose there exists a basis for $\Sdet$ such that each stabilizer generator $s \in \Sdet$ is supported on at most $\Delta$ consecutive cuts. For example, for static qLDPC codes and all known families of Floquet codes, there exists a basis such that $\Delta$ is a constant.

\begin{theorem}[Window theorem]
\label{thm:window}
Let $m=\min\{d_1(C),d_2\}$. Suppose there exists a basis of detector stabilizers in $C$ and $C_1$ such that each generator is supported on at most $\Delta$ consecutive cuts. Let $R= (m-1)(\Delta-1)$, and let
$t_1=\max\{0,\,T_1-R\}$, $t_2=\min\{T_2,\,T_1+R\}$. Define a circuit window that contains $T_1$ as $C_W=C|_{[t_1,t_2]}$. The terminal round for $C_W$ measures exactly the stabilizers at $t_2$ that $C$ reveals after $t_2$. The distance of its corresponding subsystem spacetime code is computed to be $d_W$.
Then
\begin{equation}
    d(C) = \min\{d_1(C),\ d_2,\ d_W\}.
\end{equation}
\end{theorem}

\begin{proof}
\emph{Upper bound.} The upper bound proof is trivial. Any undetectable nontrivial logical error in any of the subcircuits of $C$ must be an undetectable logical error on $C$ as well.

\emph{Lower bound.} Let $E$ be a minimum-weight logical fault of $C$, and suppose that
$w=\operatorname{wt}(E)<\min\{d_1(C),d_2,d_W\}$. Then $w\le m-1$, so $E$ is supported on at most $m-1$ cuts.

First we determine the temporal support of $E$. $E$ cannot be supported only in $C_1$ because logical faults within $C_1$ must have distance greater than $d_1(C) > w$, and similarly $E$ cannot be supported only in $C_2$. We therefore consider the case where $E$ has support on both $C_1$ and $C_2$. 

For a single fault cluster, we can confine the temporal support of $E$ that is supported on both $C_1$ and $C_2$. The fault cluster must not have $\Delta -1$ or more consecutive cuts without faults. Otherwise, we can find two smaller fault clusters each localized on a separate set of detectors, so the smaller fault clusters must be individually undetectable as well. If $E$ has non-trivial support on $t> T_1 +(\Delta-1)(m-1)$, then since each gap is at most $\Delta-2$ cuts, $E$ must be supported on $t>T_1$, giving a contradiction. The same argument works for the starting temporal support of $E$ by symmetry. Thus, $E$ must be supported temporally on 
\begin{equation*}
    I\coloneqq\big[\,T_1-(m-1)(\Delta-1),\ T_1+(m-1)(\Delta-1)\,\big].
\end{equation*}
We define the circuit window at the boundary to be $C_W \coloneqq C|_{[t_1, t_2]}$ with $t_1= T_1-(m-1)(\Delta-1),t_2 = T_1+(m-1)(\Delta-1)$. Note that the subsystem spacetime code for $C_W$ must be constructed using a partial terminal round as in Section~\ref{sec:beyond_steady_state} to avoid having extra detectors in $C_W$ as compared to in $C$. As defined here, $C_W$ has the same detectors as $C$ restricted to the time window for $C_W$. $C_W$ has the same unmasked stabilizers and destabilizers, and has the same masked components as $C$ restricted to the window.

Hence, $E$ has the same syndrome information in $C_W$ and in $C$. But $\mathrm{wt}(E)<d_W$ so $E$ must be either trivial or detectable in $C_W$. So $E$ must be trivial or detectable in $C$, as a full terminal round of measurements can unmask some permanently masked stabilizers in $C$. A contradiction. Therefore, we conclude that the equality holds, i.e. $w = \min\{d_1(C),d_2,d_W\}$.
\end{proof}

\begin{remark}[The depth assumption is needed]
The circuit of Appendix~\ref{sec: counter} has detectors spanning up to $12$ cuts. The windows $C|_{[0,t_2]}$ with $t_2\le 12$ have distance $4$.
From $t_2 =13$ on, the windows have distance $3$.
\end{remark}

\section{Conclusion}~\label{sec: conclude}
We have introduced the subsystem spacetime code of a circuit of Clifford
unitaries and Pauli measurements. Its gauge group is generated by the instantaneous gauge group, which depends on the masking properties of the elements,
together with the commuting propagators. By Theorem~\ref{thm:gauge-structure} the gauge group
consists exactly of the undetectable faults with trivial logical action. The subsystem spacetime code plays, for a
dynamical protocol, the role that the stabilizer formalism plays for a
static code: Its
center is the space of detector stabilizers, its distance is the fault
distance of the circuit, and gauge operators are exactly the harmless faults.

We use masking as the key organizing principle for defining the subsystem spacetime code. Each gauge qubit
is either a masked stabilizer paired with its
destabilizer, or an unmasked stabilizer paired with a commuting propagator, or
a pair of adjacent logical representatives. This classification helps 
explain how a circuit built from good instantaneous codes can
nevertheless have a small fault distance: A syndrome bit may never be read due to masking,
and measurement outcomes that feed the logical update may be undetectable. These two mechanisms yield two distances that bound the fault
distance from above, $d \le \min \{d_U, d_{\mathrm{meas}}\}$, where $d_U$ is the minimum
unmasked distance of all the instantaneous stabilizer groups, which is a static
computation at a fixed timeslice, and $d_{\mathrm{meas}}$ is the measurement distance, a classical
parity computation over the measurement record. For circuits that additionally satisfy a correctability condition
(see Condition 1), this bound turns out to be tight. 

One major application is circuit composition. The subsystem spacetime code framework can be used to compute the distance for each circuit component which bounds the fault-distance of the entire protocol. We show that gluing two circuits together can create logical faults that
neither piece contains. This can originate from lost detectors at the interface as temporarily
masked stabilizers become permanently masked. In general one retains only an upper bound,
$d(C) \le \min\{d_1(C), d_2, \dmeas(C)\}$. However, with an analogue of the
correctability condition (Condition 2), this bound is tight. In addition, we introduce the window
theorem which shows that the fault distance of the composed circuit equals
the minimum of the distances of the two pieces and the distance of
a window around the interface of the two smaller circuits.

The framework takes into account state preparation, logical measurement, and
reset, so it includes Floquet codes, code switching and deformation, and syndrome extraction circuits. We can design a more complicated circuit by preparing smaller gadgets and assembling them, using subsystem spacetime code to compute their fault distance or checking where the distance is lost in the composition. 

We close with several directions
that the present framework leaves open.

\subsection{Open problems}~\label{sec: open problems}
Firstly, our formalism only works for Clifford circuits and Pauli measurements. In light of the recent series of works on non-abelian codes~\cite{ZhuKobayashiHsin2026,LiWeiHanLiu2026}, it is natural to ask if the framework can be generalized to other families of quantum codes.

Secondly, the ZX calculus offers a graphical language for rewriting quantum
circuits, and distance-preserving rewrites have already been used
to construct Floquet codes \cite{RodatzPoorKissinger2026}. We imagine potential use of both subsystem spacetime code and the ZX calculus towards developing novel practical fault-tolerant architectures.

Thirdly, there exists a limit to logical computation that can be achieved using transversal gates \cite{EastinKnill2009, chakraborty2026no}. Fault-tolerant switching between codes with complementary gate sets
can overcome this limitation \cite{AndersonDuclosCianciPoulin2014}. One might expect that the subsystem spacetime code framework will be useful for developing new code switching protocols between two (or more) codes to achieve a universal set of gates.

Lastly, an example of Floquet codes without a parent subsystem code has been
constructed \cite{PRXQuantum.4.020341}, and
path-integral methods have produced new dynamical codes \cite{Bauer2024}. What general principles can
guide the discovery of further dynamical codes, including codes
outside the Pauli stabilizer setting?

\section*{Acknowledgments}
We thank Gideon Lee, Nathanan Tantivasadakarn, Joe Iverson, Xiao Xiao and Michael Gullans for valuable discussion and feedback. Part of this work was done while XF was an intern at Amazon in summer 2024. 

This material was funded in part by the U.S. Department of Energy, Office of Science, Office of Advanced Scientific Computing Research. Sandia National Laboratories is a multimission laboratory managed and operated by National Technology and Engineering Solutions of Sandia, LLC, a wholly owned subsidiary of Honeywell International Inc., for the U.S. Department of Energy’s National Nuclear Security Administration under contract DE-NA0003525. This paper describes objective technical results and analysis. Any subjective views or opinions that might be expressed in the paper do not necessarily represent the views of the U.S. Department of Energy or the United States Government. E.X.F. acknowledges NSF grant
OMA2120757 (QLCI) for Robust Quantum Simulation.

\section{AI disclosure}
Parts of this manuscript were prepared with the assistance of ChatGPT and Claude Code. These agents were used to run numerics to build intuition, to verify proofs, generate figures and to edit the writing of some sections, in accordance with the authors’ instructions. All content, claims, and conclusions have been reviewed and verified by the authors to ensure accuracy and originality. 
 
\appendix
\section{Notations}
\begingroup
\small
\setlength{\LTpre}{0.5\baselineskip}
\setlength{\LTpost}{0.5\baselineskip}
\begin{longtable}{@{}p{0.25\textwidth}p{0.27\textwidth}p{0.38\textwidth}@{}}
\caption{Notation used throughout the paper.}\label{tab:notation}\\
\toprule
Object & Notation & Meaning \\
\midrule
\endfirsthead
\multicolumn{3}{c}{\tablename~\thetable\ (continued)}\\
\toprule
Object & Notation & Meaning \\
\midrule
\endhead
\midrule
\multicolumn{3}{r}{Continued on the next page}\\
\endfoot
\bottomrule
\endlastfoot

\multicolumn{3}{@{}l}{\emph{Symplectic notation}}\\[2pt]
Symplectic form & $\symp(E,F)$ & $0$ for commuting Paulis and $1$ for anticommuting Paulis.\\
Symplectic orthogonal & $\orth{\mathsf W}$ & Normalizer of $W$.\\
Linear dual & $\dual{\mathsf W}$ & \\
\addlinespace
\multicolumn{3}{@{}l}{\emph{Circuit and spacetime notation}}\\[2pt]
Pauli group at cut $t$ & $\atcut{\Vcut}{t}$ & Phase-free (abelian) Pauli group at the time slice $t$.\\
Instantaneous stabilizer group & $\atcut{\Stab}{t}$ & Stabilizer group at cut $t$.\\
Circuit block & $\atcut{\CircuitBlock}{t+1/2}$ & Circuit between cuts $t$ and $t+1$.\\
Spacetime Pauli group & $\VST$ & $\bigotimes_{t=0}^{\FinalCut}\atcut{\Vcut}{t}$.\\

\addlinespace
\multicolumn{3}{@{}l}{\emph{Detector and logical notation}}\\[2pt]
Detector space & $\DetectorSpace$ & Binary space of detector relations.\\
Detector & $\DetElt\in\DetectorSpace$ & An individual detector relation.\\
Syndrome & $\DetResp$ & $\DetResp_\DetElt(f)$ indicates whether $f$ flips $\DetElt$. \\
Detector-to-spacetime map & $\DetMap$ & A map that sends a detector relations into $\VST$.\\
 Stabilizer from a detector & $\DetStab{\DetElt}$ & $\DetStab{\DetElt}=\DetMap(\DetElt)$.\\
Stabilizer space from detectors& $\Sdet$ & $\operatorname{im}\DetMap$.\\
 Logical space & $\Logical$ & Logical Pauli space.\\

Adjacent logical space & $\Ladj$ & Span of $\atcut{\ell}{t}+\atcut{\ell}{t+1}$, for all $t$.\\

\addlinespace
\multicolumn{3}{@{}l}{\emph{Masking and subsystem-code notation}}\\[2pt]
Sets for masking & $\atcut{\Uspace}{t},\atcut{\Tspace}{t},\atcut{\Pspace}{t}$ & Unmasked, temporarily masked, and permanently masked stabilizer subspaces at cut t.\\
Elements & $\Uelt{a}{t},\Telt{b}{t},\Pelt{c}{t}$ & Elements of the three masking subspaces.\\
Corresponding destabilizers & $\Udual{a}{t},\Tdual{b}{t},\Pdual{c}{t}$ & \\
Gauge group & $\Gspace$ & Gauge group of the subsystem spacetime code.\\
Center & $\Z{\Gspace}$ & $\Gspace\cap\orth{\Gspace}$.\\
Bare logical space & $\Lbare{\Gspace}$ & $\orth{\Gspace}\backslash\Z{\Gspace}$.\\
Dressed logical space & $\Ldressed{\Gspace}$ & $\orth{\Z{\Gspace}}\backslash\Gspace$.\\
Subsystem-code distance & $\CodeDistance{\Gspace}$ & Minimum weight of a nontrivial dressed logical Pauli.\\
Propagator spaces & $\Rprop,\Rcomm,\Ranti$ & All, measurement-commuting, and measurement-anticommuting propagators. \\
\end{longtable}
\endgroup
\section{Example where Condition 1 fails}~\label{sec: counter}
The following randomly generated 12-qubit code provides a counterexample to Condition 1. It illustrates that for a quantum code where unmasked stabilizers are constantly initialized and destroyed there may be no relation between the unmasked distance and the distance of the subsystem spacetime code. This is true even if the logical representatives do not dynamically evolve.

\begin{align}
  \Stab^{(0)} &= \langle g_1,\dots,g_{11}\rangle.
\end{align}

The code encodes one logical qubit, with logical operators
\begin{equation}
  \bar{X} = \texttt{ZYZYYIYIXIYY}, \qquad
  \bar{Z} = \texttt{YXXYZYXZZZXI},
\end{equation}
both of which commute with every measurement in the measurement scheme. The generators of $\Stab^{(0)}$ are given in Table~\ref{tab:initial-isg} and the measurement scheme, along with the effects of each measurement (either unmasking or masking a stabilizer), is described in Table~\ref{tab:schedule}.

\begin{table}[ht]
\centering
\begin{subtable}[t]{0.48\textwidth}
\centering
\caption{}
\label{tab:initial-isg}
\begin{tabular}{@{}lc@{}}
\toprule
name & Pauli \\
\midrule
$g_1$           & \texttt{IYYIYYIXYXXZ} \\
$g_2 = h$       & \texttt{YZXXZZYYXXYX} \\
$g_3$           & \texttt{YYIIXIIYZZZY} \\
$g_4$           & \texttt{YYXXZXXXZXZX} \\
$g_5$           & \texttt{XXIIIZYIZZZY} \\
$g_6$           & \texttt{YZYYZYIXXXYI} \\
$g_7$           & \texttt{IXXXIIIXYIYY} \\
$g_8$           & \texttt{IZXIXZIYXXIZ} \\
$g_9$           & \texttt{YIXZYIIIIZIZ} \\
$g_{10}$        & \texttt{ZIIIXZZXZIYZ} \\
$g_{11}$        & \texttt{YIIIIYZXZZIZ} \\
\bottomrule
\end{tabular}
\end{subtable}
\hfill
\begin{subtable}[t]{0.48\textwidth}
\centering
\caption{}
\label{tab:measured-operators}
\begin{tabular}{@{}lc@{}}
\toprule
name & Pauli \\
\midrule
$M_1 = c_1$ & \texttt{XXXZZXZIZXZY} \\
$c_2$       & \texttt{IXXXYZYIYYXZ} \\
$c_3$       & \texttt{IZYIYIZYIXIZ} \\
$c_4$       & \texttt{ZIZXXYIZIZII} \\
$c_5$       & \texttt{IYXZYZYYYXXY} \\
$c_6$       & \texttt{IZZXYYIIIXZX} \\
$c_7$       & \texttt{IXZIZXIZZIXI} \\
$c_8$       & \texttt{IZIYXZYYXYYY} \\
$c_9$       & \texttt{ZYYIZXZIXXZI} \\
$c_{10}$    & \texttt{IZXXZXXZYYYY} \\
$h$         & \texttt{YZXXZZYYXXYX} \\
$M_2$       & \texttt{IIXYIYYXZXZX} \\
\bottomrule
\end{tabular}
\end{subtable}

\caption{Pauli operators defining the 12-qubit example.
(a) Generators of the initial ISG $\Stab^{(0)}$.
(b) Operators appearing in the measurement scheme.}
\label{tab:example-operators}
\end{table}

\begin{table}[ht]
\centering
\begin{tabular}{@{}llllp{4.2cm}@{}}
\toprule
$m_i$ & Operator measured & Effect \\
\midrule
$m_0$ & $M_1$ & masks $g_1$; $\Stab^{(t_0)}\to \Stab^{(t_1)}$ \\
$m_1$ & $c_1\,(=M_1)$ & detector \\
$m_2,\dots,m_{10}$ & $c_2,\dots,c_{10}$ & detectors \\
$m_{11}$ & $h$ & detector; the only readout of $h$ \\
$m_{12}$ & $M_2$ & masks $h$; $\Stab^{(t_1)}\to \Stab^{(t_2)}$ \\
$m_{13},\dots,m_{22}$ & $c_1,\dots,c_{10}$ & detectors \\
$m_{23}$ & $M_2$ & detector \\
Periodic measurements after $t=24$ & $\{c_1,\dots,c_{10},M_2\}$ & detectors \\
\bottomrule
\end{tabular}
\caption{Measurement schedule with one measurement per timestep. Measurement
$m_i$ occurs at time $i+0.5$}
\label{tab:schedule}
\end{table}

The ISG at each time step is given by,
\begin{align}
  \Stab^{(t_1)} &= \langle c_1,c_2,\dots,c_{10},h\rangle
    && t_1 \in \{1\ldots 12\}, \\
  \Stab^{(t_2)} &= \langle c_1,\dots,c_{10},M_2\rangle
    && t_2 \in \{13\ldots 24\}.
\end{align}

The unmasked distance for each ISG is verified to be $4$ but the distance as a subsystem spacetime code is 3 with measurement error treated as a weight 2 Pauli error. A weight-3 spacetime logical is given by $(Y_1^{(0)} + Y_1^{(2)} + Z_1^{(12)})$ which is equivalent to a measurement error and a Pauli error each individually detectable.

\clearpage

\bibliography{references.bib}

\newcommand{\etalchar}[1]{$^{#1}$}
\begin{thebibliography}{MdlFOTT{\etalchar{+}}25}

\bibitem[AB25]{aitchison2025spacetimespinsstatisticalmechanics}
Cory~T. Aitchison and Benjamin Béri.
\newblock Spacetime spins: Statistical mechanics for error correction with
  stabilizer circuits, 2025.
\newblock arXiv:2512.21991.

\bibitem[ABO08]{aharonov1999faulttolerant}
Dorit Aharonov and Michael Ben-Or.
\newblock Fault-tolerant quantum computation with constant error rate.
\newblock {\em SIAM Journal on Computing}, 38(4):1207--1282, 2008.
\newblock arXiv:quant-ph/9906129.

\bibitem[ADCP14]{AndersonDuclosCianciPoulin2014}
Jonas~T. Anderson, Guillaume Duclos-Cianci, and David Poulin.
\newblock Fault-tolerant conversion between the {Steane} and {Reed--Muller}
  quantum codes.
\newblock {\em Physical Review Letters}, 113:080501, 2014.
\newblock arXiv:1403.2734.

\bibitem[AGP06]{aliferis2006threshold}
Panos Aliferis, Daniel Gottesman, and John Preskill.
\newblock Quantum accuracy threshold for concatenated distance-3 codes.
\newblock {\em Quantum Information and Computation}, 6(2):97--165, 2006.
\newblock arXiv:quant-ph/0504218.

\bibitem[AR25]{Anonymous_2025}
M.~Sohaib Alam and Eleanor Rieffel.
\newblock Dynamical logical qubits in the {Bacon-Shor} code.
\newblock {\em Physical Review A}, July 2025.
\newblock arXiv:2403.03291.

\bibitem[Bau24]{Bauer2024}
Andreas Bauer.
\newblock Topological error correcting processes from fixed-point path
  integrals.
\newblock {\em Quantum}, 8:1288, 2024.
\newblock arXiv:2303.16405.

\bibitem[Bau25]{PhysRevA.111.032413}
Andreas Bauer.
\newblock $x+y$ {Floquet} code: A simple example for topological quantum
  computation in the path-integral approach.
\newblock {\em Phys. Rev. A}, 111:032413, Mar 2025.

\bibitem[BFHS17]{bacon_sparse_2017}
Dave Bacon, Steven~T. Flammia, Aram~W. Harrow, and Jonathan Shi.
\newblock Sparse {Quantum} {Codes} from {Quantum} {Circuits}.
\newblock {\em IEEE Transactions on Information Theory}, 63(4):2464--2479,
  April 2017.
\newblock arXiv:1411.3334.

\bibitem[BH25]{blackwell2025code}
Keller Blackwell and Jeongwan Haah.
\newblock The code distance of {Floquet} codes.
\newblock {\em arXiv preprint arXiv:2510.05549}, 2025.
\newblock arXiv:2510.05549.

\bibitem[BLN{\etalchar{+}}24]{Bombin2024unifyingflavorsof}
Hector Bombin, Daniel Litinski, Naomi Nickerson, Fernando Pastawski, and Sam
  Roberts.
\newblock Unifying flavors of fault tolerance with the {ZX} calculus.
\newblock {\em {Quantum}}, 8:1379, June 2024.

\bibitem[CG26]{chakraborty2026no}
Aranya Chakraborty and Daniel Gottesman.
\newblock No-go theorem on fault tolerant gadgets for multiple logical qubits.
\newblock {\em PRX Quantum}, 7(3):033027, 2026.

\bibitem[DKLP02]{dennis2002topological}
Eric Dennis, Alexei Kitaev, Andrew Landahl, and John Preskill.
\newblock Topological quantum memory.
\newblock {\em Journal of Mathematical Physics}, 43:4452--4505, 2002.
\newblock arXiv:quant-ph/0110143.

\bibitem[dlF24]{delafuente2024dynamicalweightreductionpauli}
Julio C.~Magdalena de~la Fuente.
\newblock Dynamical weight reduction of {Pauli} measurements, 2024.
\newblock arXiv:2410.12527.

\bibitem[DP23]{delfosse2023spacetimecodescliffordcircuits}
Nicolas Delfosse and Adam Paetznick.
\newblock Spacetime codes of {Clifford} circuits, 2023.
\newblock arXiv:2304.05943.

\bibitem[DTB23]{PRXQuantum.4.020341}
Margarita Davydova, Nathanan Tantivasadakarn, and Shankar Balasubramanian.
\newblock {Floquet} codes without parent subsystem codes.
\newblock {\em PRX Quantum}, 4:020341, Jun 2023.

\bibitem[DTBA24]{Davydova2024quantumcomputation}
Margarita Davydova, Nathanan Tantivasadakarn, Shankar Balasubramanian, and
  David Aasen.
\newblock Quantum computation from dynamic automorphism codes.
\newblock {\em {Quantum}}, 8:1448, August 2024.

\bibitem[DTSE24]{PRXQuantum.5.020305}
Arpit Dua, Nathanan Tantivasadakarn, Joseph Sullivan, and Tyler~D. Ellison.
\newblock Engineering {3D} {Floquet} codes by rewinding.
\newblock {\em PRX Quantum}, 5:020305, Apr 2024.

\bibitem[EK09]{EastinKnill2009}
Bryan Eastin and Emanuel Knill.
\newblock Restrictions on transversal encoded quantum gate sets.
\newblock {\em Physical Review Letters}, 102:110502, 2009.
\newblock arXiv:0811.4262.

\bibitem[FG25]{fu2024errorcorrectiondynamicalcodes}
Esther~Xiaozhen Fu and Daniel Gottesman.
\newblock Error correction in dynamical codes.
\newblock {\em {Quantum}}, 9:1886, October 2025.
\newblock arXiv:2403.04163.

\bibitem[Gid23]{Gidney_2023}
Craig Gidney.
\newblock A pair measurement surface code on pentagons.
\newblock {\em Quantum}, 7:1156, October 2023.

\bibitem[Got97]{gottesman1997stabilizer}
Daniel Gottesman.
\newblock {\em Stabilizer Codes and Quantum Error Correction}.
\newblock PhD thesis, California Institute of Technology, 1997.
\newblock arXiv:quant-ph/9705052.

\bibitem[Got22]{gottesman2022opportunities}
Daniel Gottesman.
\newblock Opportunities and challenges in fault-tolerant quantum computation,
  2022.
\newblock arXiv:2210.15844. To appear in Proc. 28th Solvay Conference on
  Physics.

\bibitem[HB24]{PRXQuantum.5.040327}
Oscar Higgott and Nikolas~P. Breuckmann.
\newblock Constructions and performance of hyperbolic and semi-hyperbolic
  {Floquet} codes.
\newblock {\em PRX Quantum}, 5:040327, Nov 2024.

\bibitem[HH21]{Hastings2021dynamically}
Matthew~B. Hastings and Jeongwan Haah.
\newblock Dynamically {G}enerated {L}ogical {Q}ubits.
\newblock {\em {Quantum}}, 5:564, October 2021.

\bibitem[KMdlFT{\etalchar{+}}24]{PRXQuantum.5.010342}
Markus~S. Kesselring, Julio~C. Magdalena de~la Fuente, Felix Thomsen, Jens
  Eisert, Stephen~D. Bartlett, and Benjamin~J. Brown.
\newblock Anyon condensation and the color code.
\newblock {\em PRX Quantum}, 5:010342, Mar 2024.

\bibitem[LWHL26]{LiWeiHanLiu2026}
Zimu Li, Fuchuan Wei, Zhengyi Han, and Zi-Wen Liu.
\newblock {Non-Abelian} sheaf quantum {LDPC} codes: good and magical, 2026.
\newblock arXiv:2609.30159.

\bibitem[MBG23]{McEwen_2023}
Matt McEwen, Dave Bacon, and Craig Gidney.
\newblock Relaxing hardware requirements for surface code circuits using
  time-dynamics.
\newblock {\em Quantum}, 7:1172, November 2023.

\bibitem[MdlFOTT{\etalchar{+}}25]{PRXQuantum.6.010360}
Julio~C. Magdalena de~la Fuente, Josias Old, Alex Townsend-Teague, Manuel
  Rispler, Jens Eisert, and Markus M\"uller.
\newblock $\mathrm{XYZ}$ ruby code: Making a case for a three-colored graphical
  calculus for quantum error correction in spacetime.
\newblock {\em PRX Quantum}, 6:010360, Mar 2025.

\bibitem[RPK26]{RodatzPoorKissinger2026}
Benjamin Rodatz, Boldizs{\'a}r Po{\'o}r, and Aleks Kissinger.
\newblock {Floquetifying} stabiliser codes with distance-preserving rewrites.
\newblock {\em Quantum}, 10:2202, 2026.
\newblock arXiv:2410.17240.

\bibitem[ST25]{PhysRevLett.134.090602}
Mackenzie~H. Shaw and Barbara~M. Terhal.
\newblock Lowering connectivity requirements for bivariate bicycle codes using
  morphing circuits.
\newblock {\em Phys. Rev. Lett.}, 134:090602, Mar 2025.

\bibitem[SWP23]{PhysRevB.108.195134}
Joseph Sullivan, Rui Wen, and Andrew~C. Potter.
\newblock {Floquet} codes and phases in twist-defect networks.
\newblock {\em Phys. Rev. B}, 108:195134, Nov 2023.

\bibitem[TGB25]{tanggara2025simpleconstructionquditfloquet}
Andrew Tanggara, Mile Gu, and Kishor Bharti.
\newblock Simple construction of qudit {Floquet} codes on a family of lattices,
  2025.
\newblock arXiv:2410.02022.

\bibitem[TTMdlFK23]{Townsend_Teague_2023}
Alex Townsend-Teague, Julio Magdalena de~la Fuente, and Markus Kesselring.
\newblock {Floquetifying} the colour code.
\newblock {\em Electronic Proceedings in Theoretical Computer Science},
  384:265--303, August 2023.

\bibitem[Vui21]{vuillot2021planarfloquetcodes}
Christophe Vuillot.
\newblock Planar {Floquet} codes, 2021.
\newblock arXiv:2110.05348.

\bibitem[XD25]{xu2025faulttolerantprotocolsspacetimeconcatenation}
Yichen Xu and Arpit Dua.
\newblock Fault-tolerant protocols through spacetime concatenation, 2025.
\newblock arXiv:2504.08918.

\bibitem[XWL26]{xu2026theoryspacetimequantumfault}
Yijia Xu, Yixu Wang, and Zi-Wen Liu.
\newblock Theory of spacetime quantum fault tolerance, 2026.
\newblock arXiv:2609.39758.

\bibitem[YCC24]{yan2024floquetcodescoupledspin}
Bowen Yan, Penghua Chen, and Shawn~X. Cui.
\newblock {Floquet} codes from coupled spin chains, 2024.
\newblock arXiv:2410.18265.

\bibitem[ZAV23]{PhysRevB.108.205116}
Zhehao Zhang, David Aasen, and Sagar Vijay.
\newblock $x$-cube {Floquet} code: A dynamical quantum error correcting code
  with a subextensive number of logical qubits.
\newblock {\em Phys. Rev. B}, 108:205116, Nov 2023.

\bibitem[ZKH26]{ZhuKobayashiHsin2026}
Guanyu Zhu, Ryohei Kobayashi, and Po-Shen Hsin.
\newblock {Non-Abelian qLDPC}: {TQFT} formalism, addressable gauging
  measurement and application to magic state fountain on {2D} product codes,
  2026.
\newblock arXiv:2601.06736.

\end{thebibliography}
\end{document}